\documentclass[aps,prd,nofootinbib,english,10pt]{revtex4-2}
\usepackage[utf8]{inputenc}

\usepackage{graphicx,color} 
\usepackage{xcolor}
\usepackage{rotating}
\usepackage{amssymb}
\usepackage{amsfonts}
\usepackage{amsmath}
\usepackage{amsthm}
\usepackage{xspace}
\usepackage{mathtools} 
\usepackage{verbatim}
\usepackage{comment}
\usepackage{enumitem}
\usepackage{physics}
\usepackage{graphicx}
\usepackage[colorlinks]{hyperref}
\usepackage{MnSymbol}
\usepackage{graphicx}
\usepackage{accents}
\usepackage[mathscr]{euscript}

\newtheorem{theorem}{Theorem}
\newtheorem{definition}{Definition}
\newtheorem{remark}{Remark}
\newtheorem{lemma}{Lemma}
\newtheorem{proposition}{Proposition}
\newtheorem{corollary}{Corollary}

\newcommand{\y}{\dot{x}}

\newcommand{\vd}{\dot{\partial}}
\newcommand{\Dd}{\mathrm{d}}

\begin{document}

\title{Reduction of the Finsler gravity vacuum equation and dynamics for the cosmological Landsberg spacetimes}

\author{Christian Pfeifer}
\email{christian.pfeifer@zarm.uni-bremen.de}
\affiliation{ZARM, University of Bremen, 28359 Bremen, Germany.}

\author{Fidel F. Villaseñor}
\email{fidel.fernandez@uclm.es}
\affiliation{Department of Mathematics, University of Castilla-La Mancha, 13001 Ciudad Real, Spain}

\begin{abstract}
When solving the Einstein vacuum equations, a very helpful feature is that they reduce simply to the vanishing of the Ricci tensor. In Finsler gravity, a promising extension of general relativity that can describe the gravitational field of kinetic gases from a phase space perspective in terms of Finsler geometry, such a reduction is not as straightforward.

In this article, we identify precise conditions under which the scalar Finsler gravity vacuum equation (in either its purely metric or its Palatini formulation) reduces to the vanishing of the Finslerian Ricci curvature. Through analytic arguments, we find that this happens if there exists some power $F^n$ of the Finsler function $F$ that is sufficiently regular and whose associated Finsler metric is non-degenerate on the light cones. Moreover, the Landsberg term in the scalar equation must vanish. This result significantly generalizes the findings of \cite{villasenor:2024}, where a reduction theorem was established under the quite strong  assumption that $F^2$ is regular, which is not satisfied by many examples currently under consideration.

We demonstrate the impact of our findings by applying them to solve the Finsler gravity equations for homogeneous and isotropic Finsler spacetime functions of Landsberg type.
\end{abstract}

\maketitle

%\tableofcontents

%%%%%%%%%%%%%%%%%%%%%%%%%%%%%%%%%%%%%%%%%%%%%%%%%%%%%%%%%%%%%%%%%%%%%%%%%%%%%%%%
%%%%%%%%%%%%%%%%%%%%%%%%%%%%%%%%%%%%%%%%%%%%%%%%%%%%%%%%%%%%%%%%%%%%%%%%%%%%%%%%
\section{Introduction}
To describe physical systems from field theory, it is essential to solve the corresponding field equations, which, generically, are highly non-trivial coupled non-linear partial differential equations. Every simplification of the field equations is welcome in order to be able to find analytic solutions, to improve our understanding and predictions for the behavior of the physical system. 

In gravitational physics, the most famous and most successful field equations that describe the gravitational field of an energy-momentum distribution on spacetime are the Einstein equations. They are non-linear second order partial differential equations which determine the spacetime metric that, in turn, determines the geometry of spacetime. In vacuum, the Einstein equations reduce to the vanishing of the Ricci tensor, which simplifies the search for solutions tremendously.

Even though general relativity is enormously successful, it is clear that, together with the standard model of particle physics, it cannot be the final answer for the description of gravity. This fact becomes very clear by the necessity to introduce dark matter and dark energy to match all observations \cite{Leauthaud:2025azz,Capozziello:2025qmh,Devlin:2024gdl,Bertone:2016nfn} (which becomes even more difficult due to the tensions in the determination of the Hubble parameter $H_0$ \cite{Abdalla:2022yfr,CosmoVerseNetwork:2025alb}). Moreover, from the theoretical side, there is the inevitable prediction of singularities \cite{Penrose1965,HawkingPenrose1970} or infinite gravitational tidal force, or the missing theory of quantum gravity \cite{Addazi:2021xuf,Buoninfante:2024yth}.

For gravity theories beyond general relativity, typically the field equations  are even more complicated than the Einstein ones, and so is the search for solutions. A very promising such theory is Finsler gravity \cite{Beem,Asanov1985,Rutz1993FinslerEinstein,Kouretsis:2008ha,ChangLi2009,java-sanchez:2014,Stavrinos2014,Caponio:2015hca,minguzzi_2016,Javaloyes:2018lex,Lammerzahl:2018lhw,villasenor:2022,Heefer:2024kfi,Sanchez:2025ifx}. It describes the geometry of spacetime in terms of a general length measure for curves \cite{BCS,Bucataru,Sanchez:2025ifx,Hohmann:2021zbt}; physically speaking, a Lagrangian for freely falling point particles, instead of a spacetime metric. 

The Finsler gravitational field equations that are discussed in the literature can be divided into two classes. 1) Tensorial type, formally similar to the Einstein equations and which determine the so-called Finsler metric, usually sourced by some kind of generalized energy momentum tensor \cite{Kouretsis:2008ha,GarciaParrado2022,Bucataru}. 2) Scalar type, for which the fundamental variable is the length measure-defining Finsler function \cite{Rutz1993FinslerEinstein,HPV19} and which can be sourced directly by the $1$-particle distribution function of a kinetic gas \cite{Hohmann:2020yia,Pfeiferetal:2020}. This second category has also been generalized to a Finslerian extension of Palatini gravity, in which the scalar equation is supplemented by a $(1,1)$-tensor field equation determining an independent non-linear connection as another fundamental variable \cite{villasenor:2022}.

Our viewpoint is that there is better physical motivation for the second class of Finsler gravity theories: the fundamental geometry-describing variables are identical to those to be determined dynamically from the field equations. Moreover, these equations are derived from a well-defined action principle, that includes a canonical coupling between Finsler spacetime geometry and kinetic gas matter, which has been shown to be quite successful. Indeed, it offers a possible explanation for the accelerated expansion of the universe without the need for dark energy \cite{Pfeifer_2025}.

The scalar Finsler gravity equations are difficult to solve, as the limited number of known exact solutions demonstrates \cite[Ch. III]{Heefer:2024kfi}. One difficulty lies in the fact that even for the vacuum equations, no reduction to a simple expression, like the vanishing of the Ricci tensor in general relativity, is known. \vspace{10pt}

In this article, we analyze the conditions under which the scalar Finsler gravity vacuum equations reduce to the vanishing of the Finslerian Ricci curvature. The latter is the Finslerian Ricci-flatness condition, the equation that was originally proposed by S. Rutz in~\cite{Rutz1993FinslerEinstein}. It will turn out that this reduction is way less trivial than in the case of the Einstein vacuum equations, but it can be understood in a precise way.\vspace{10pt}

Our main finding will be that the aforementioned reduction
is possible in case that the so-called Landsberg term of the equation vanishes, if there exists a certain power $F^n$ of the Finsler function $F$ that is regular enough on the light cones of the Finsler spacetime.  This means that $F^n$ is differentiable a sufficient number of times (say, smooth) at the light cones and the \emph{$n$-Finsler metric} is non-degenerate there. (Precise definitions and hypotheses will follow in the main body of the paper.)  In the future, these results will help to find explicit solutions of the Finsler gravity equations in a much simpler way than it is possible now. \vspace{10pt}

To demonstrate our claims, this article is organized as follows. In Section \ref{sec:FST}, we recall the mathematical basics about Finsler spacetimes and Finsler gravity in both the metric and the Palatini formulations; this sets the stage for our ongoing analysis. Afterwards, in Section \ref{sec:F^n}, we show how it is possible to express the whole geometry of a Finsler spacetime in terms of $F^n$ instead of $F^2$. Arbitrary powers $n$ of the Finsler function $F$ are considered as they turn out to be useful, so we argue that one should not just stick to the $n=2$ case, as is usually done in the literature on Finsler geometry. The main results on the reduction of the vacuum Finsler gravity equations are then presented in Section \ref{sec:RedVacEq}. Interestingly, we find two conditions under which the reduction is  possible; we discuss them with the necessary theoretical detail. In Section \ref{eq:uniex}, after clarifying the situation of Berwald solutions, we solve the Finsler gravity equation  from its reduced version for the so-called cosmological unicorn, as an explicit example of the scope of our results. Last but not least, we summarize and conclude our  findings in Section \ref{sec:conc}. In addition, there is an Appendix \ref{app:TechDet},  devoted to technicalities  on mathematical analysis and topology.

%%%%%%%%%%%%%%%%%%%%%%%%%%%%%%%%%%%%%%%%%%%%%%%%%%%%%%%%%%%%%%%%%%%%%%%%%%%%%%%%
%%%%%%%%%%%%%%%%%%%%%%%%%%%%%%%%%%%%%%%%%%%%%%%%%%%%%%%%%%%%%%%%%%%%%%%%%%%%%%%%
\section{Finsler spacetimes and Finsler gravity}\label{sec:FST}

In this section, we clarify the mathematical required setup of Finsler spacetimes and the Finsler gravity equation. This lays the foundation on which our main Theorems \ref{thm:maximumprinciplemain} to \ref{thm:analyticity main} about the reduction of the Finsler gravity equation are based in Section \ref{sec:RedVacEq}. 

We present all our findings on a (smooth, connected) 4-dimensional manifold $M$, serving as physical spacetime, and its tangent bundle $TM$. (Extensions to arbitrary dimensions are possible but will not be discussed here.) 

We use the Einstein convention, where contraction of an upper and a lower index means summation over the index pair; Greek indices are $\mu,\nu,\rho,\sigma\in\left\{0,\ldots,3\right\}$ and Latin ones are $i,j,k\in\left\{1,2,3\right\}$. We express elements $(x,\y)\in TM$ by natural coordinates $(x^\mu,\y^\mu)$, writing the tangent vector as $\y=\y^\mu\left.\partial_{x^\mu}\right|_x\in T_xM$. Given a function $f$ on $TM$, we write 
\begin{align}
    \partial_\mu f  := \frac{\partial f}{\partial x^\mu} \,,\quad 
f_{\cdot \mu}:=\vd_{\mu}f:=\frac{\partial f}{\partial\y^\mu}\,
\end{align}
for its derivatives. Therefore, the coordinate basis elements of the tangent and cotangent spaces to $TM$ are denoted by $\{\partial_\mu, \dot\partial_\mu\}$ and $\{dx^\mu, d\dot x^\mu\}$. Throughout the article, Euler's theorem will be in use: when regular enough, a function $f$ is (positively) $n$-homogeneous with respect to $\y$ if and only if $f_{\cdot \mu}\y^\mu=nf$. Analytical properties of functions and topological properties of their domains will be important, but they are delegated to the technical Appendix \ref{app:TechDet}.

\subsection{Finsler spacetimes}
To define a Finsler spacetime, we consider functions $F:\mathcal{T}\rightarrow \mathbb{R}$, where $\mathcal{T}\subseteq TM\setminus\mathbf{0}$ is open and conic, meaning~($(x,\y)\in\mathcal{T}\Rightarrow(x,\lambda\y)\in\mathcal{T}$ $\forall \lambda >0$), and has all of its fibers $\mathcal{T}_x:=\mathcal{T}\cap T_xM$ nonenmpty. Then, we state under which conditions some $F$ qualifies as a Finsler spacetime structure.

\begin{definition} \label{finsler lagrangian def}
    A \emph{Finsler spacetime function} is a smooth function $F\colon \mathcal{T}\rightarrow\mathbb{R}$ with the following properties: 
    \begin{enumerate}
        
        \item It is \emph{$1$-homogeneous}, $F(x,\lambda \y)=\lambda F(x,\y)$ $\forall(x,\y)\in \mathcal{T},\lambda>0$.
        
        \item \label{finsler lagrangian def 3} At each $x\in M$, let the \emph{indicatrix} of $F$ be $\Sigma_x:=\left\{\y\in\mathcal{T}_x\colon F(x,\y)=1\right\}$; on each indicatrix, the (vertical) Hessian 
        \begin{align}
            F_{\cdot \mu\cdot \nu} = \frac{\partial }{\partial\y^\mu}\frac{\partial }{\partial\y^\nu}F\,
        \end{align} 
        is negative definite. In other words, every time that a vector $\xi = \xi^\mu \dot\partial_\mu$ is tangent to $\Sigma_x$, meaning that $F(x,\y)=1$ and $\xi^\mu F_{\cdot \mu}(x,\y)=0$, one has 
        \begin{align}
            F_{\cdot \mu\cdot \nu}(x,\y)\,\xi^\mu \xi^\nu<0.
        \end{align}
        \item \label{finsler lagrangian def 1} Each $\mathcal{T}_x$ is connected and $F>0$ in all of $\mathcal{T}$.

        \item $F$ extends continuously as $0$ to the boundary $\partial\mathcal{T}$ of $\mathcal{T}$. This means that at each $x\in M$, we require
        \begin{align}
            F(x,\dot x) = 0\;\;\; \forall \dot x\in\left(\partial\mathcal{T}\right)_x\subset T_xM \setminus\left\{0\right\}.
        \end{align}
    \end{enumerate}
\end{definition}
Taking $L=F^2$ and focusing on\footnote{It is is the same as $\mathcal{T}$ in \cite[Def. 1]{Hohmann:2021zbt} or \cite[\S 2]{Pfeifer_2025}, while in \cite{villasenor:2022} this set had been called $A$.} 
$\mathcal{T}$, our definition here agrees with the one presented in \cite{Hohmann:2021zbt}. However, it differs from other definitions in the literature 
\cite{java-sanchez:2020,Beem,Caponio:2015hca,Beem,minguzzi_2016,Lammerzahl:2012,java-sanchez:2014}. These different definitions reflect different focuses of applications of Finsler spacetimes by different authors.

Moreover, our definition can equivalently be formulated in terms $L$: for example, the requirement \ref{finsler lagrangian def 3}. above is completely equivalent to the often used condition: 
\begin{equation}
    \forall(x,\y)\in \mathcal{T},\;\left(F^2\right)_{\cdot \mu\cdot \nu}(x,\y)\hbox{ \rm has Lorentzian signature }(+,-,-,-).
    \label{lorentzian signature}
\end{equation}
Yet, we will argue that in general $F^2$ is not special compared to other powers $F^n$ of the Finsler function and that, rather, these other powers can have advantages in determining the geometry of a Finsler spacetime. From this point of view, formulating the definition in terms of $F$ itself becomes more convenient.

We did not specify how regular, i.e.\ how often $F$, or some power of $F$, should be differentiable along the null directions. Since there is no general consensus about this point in the literature, and we need a certain regularity for our reduction results in Section \ref{sec:RedVacEq}, we postpone this discussion to Section \ref{ssec:lightlike regularity}.

Definition \ref{finsler lagrangian def} guarantees the existence of all necessary features, in order to use Finsler spacetimes as a generalization of pseudo-Riemannian (Lorentzian) spacetimes as the stage on which physics takes place:
\begin{enumerate}
    \item ensures the reparameterization invariance of the length integral for curves $\gamma$, 
    \begin{align}
        \mathscr{S}[\gamma] = \int d\tau\ F(\gamma, \dot \gamma)\,,
    \end{align} 
    which has the well-known interpretation of proper time along worldlines in $M$. 

    \item
    guarantees that the conical sets $\mathcal{T}_x$ are convex. They are interpreted as the set of \emph{timelike vectors}, and the sets $\Sigma_x$ are the sets of \emph{unit timelike vectors} at $x\in M$. They can be understood as the Finsler generalizations of the usual hyperboloids that are known from special and general relativity.

    \item states that among the components that might constitute the timelike cones at $x$, one is selected and identified as future pointing; this constitutes a smooth time orientation. It also says that no hypothetical lightlike or spacelike vector has been included in this domain; consequently, the pair $(\mathcal{T},F)$ contains exactly the same information as $\Sigma:=\underset{x\in M}{\bigcup}\Sigma_x$.

    Thus, $\mathcal{T}$ is the bundle of future pointing timelike vectors over $M$. Most importantly, on $\mathcal{T}$, the Finsler gravity action \eqref{eq:gravact} and the resulting vacuum equation \eqref{eq:FinsGravL} (alternatively, \eqref{eq:PalatR} and \eqref{eq:PalatN}) are well defined and we can prove the reduction Theorems~\ref{thm:maximumprinciplemain} to~\ref{thm:analyticity_palatini}.

    \item ensures the existence of non-trivial null (or lightlike) directions that are the boundary of the timelike directions, the so-called light cone. In general, these might not be all null directions that exist. For example, in the case of spacetimes with multiple light cones, which describe the phenomenon of birefringence \cite{Pfeifer:2019wus}, further null directions exist. However, the null directions in $\partial\mathcal{T}$ are the ones that must exist at least, in order to study spacetime physics on the basis of Finsler geometry. 
    
    The set of causal directions is then given by the closure $\overline{\mathcal{T}}=\mathcal{T}\cup\partial\mathcal{T}$ (up to possible further null directions that are only connected to $\partial \mathcal{T}$ through the $\mathbf{0}$ section). Physically speaking these are all future pointing timelike and lightlike directions, the directions along which massive and massless particles propagate. This far, most, if not all, definitions of Finsler spacetimes in the literature agree.
\end{enumerate}
All in all, formally, a \emph{Finsler spacetime} $(M,F)$ is a 4-dimensional manifold $M$ equipped with a Finsler spacetime function in the sense of Definition \ref{finsler lagrangian def}, which ensures the well-definedness of basic causality notions to do physics.

\subsection{Geometric objects on Finsler spacetimes}\label{ssec:geomobj}
To describe the dynamics of Finsler spacetime through action-based Finsler gravity (and Finsler Palatini gravity), we quickly recall the most important classical geometric objects in Finsler geometry \cite{BCS,Bucataru,Hohmann:2021zbt, Heefer:2024kfi}. Historically and conventionally, they are derived from the square of the Finsler function, namely $L=F^2$.

\begin{itemize}
    \item The Finsler metric $g^{(2)}{}_{\mu\nu}$, the Cartan tensor $C_{\mu\nu\rho}$ and its trace $C_\mu$:
    \begin{align}
        g^{(2)}{}_{\mu\nu} = \frac{1}{2}\left(F^2\right)_{\cdot\mu\cdot\nu},\quad 
        C_{\mu\nu\rho} = \frac{1}{4}\left(F^2\right)_{\cdot\mu\cdot\nu\cdot\rho},\quad 
        C_\mu = g^{(2)\nu\rho}C_{\mu\nu\rho}\,.
    \end{align}
    \item The geodesic spray $G^\mu$ and the canonical Cartan non-linear connection coefficients $N^\mu{}_\nu$:
    \begin{align}
        G^\mu = \frac{1}{4}g^{(2)}{}^{\mu\nu}\left(\dot x^\sigma \partial_\sigma \dot \partial_\nu F^2 - \partial_\nu F^2\right),\quad N^\mu{}_\nu = \dot\partial_\nu G^\mu\,.
    \end{align}
    \item The horizontal derivative, the non-linear curvature and the Finsler Ricci curvature, respectively:
    \begin{equation}
        \delta_\mu = \partial_\mu - N^\sigma{}_\mu \dot\partial_\sigma,\quad 
        R^\sigma{}_{\nu\mu}\dot \partial_\sigma = [\delta_\mu,\delta_\nu] \,,\quad
        \mathcal{R} = R^\sigma{}_{\sigma\mu}\dot x^\mu\,.
        \label{eq:def_curvatures}
    \end{equation}
    \item The dynamical covariant derivative operators $\nabla$, acting on $(r,s)$-$d$-tensor fields:
    \begin{align}
        \nabla T^{\mu_1..\mu_r}{}_{\nu_1..\nu_s} = \dot x^\sigma \delta_\sigma T^{\mu_1..\mu_r}{}_{\nu_1..\nu_s} + N^{\mu_1}{}_{\sigma}T^{\sigma..\mu_r}{}_{\nu_1..\nu_s} + ... - N^{\sigma}{}_{\nu_s}T^{\mu_1 ..\mu_r}{}_{\nu_1..\sigma}\,.
    \end{align}
    \item The Landsberg tensor $P_{\mu\nu\rho}$ and its trace $P_\mu$:
    \begin{align}
         P_{\mu\nu\rho} = \nabla C_{\mu\nu\rho},\quad P_\mu = g^{(2)}{}^{\rho\nu}P_{\mu\nu\rho}\,.
    \end{align}
\end{itemize}
For Finsler spacetime functions $F$ that are induced by a pseudo-Riemannian metric with components $g_{\mu\nu}(x)$, namely $F_{\rm pr}(x,\y)=\sqrt{g_{\mu\nu}(x)\,\dot x^\mu \dot x^\nu}$ with signature $(+,-,-,-)$, the geometric objects reduce to the ones known in pseudo-Riemannian geometry and the Finsler objects $C_{\mu\nu\rho}$ and $P_{\mu\nu\rho}$ become identically $0$. This is because the geodesic spray, the Cartan non-linear connection and the non-linear curvature are straightforward generalizations of the Levi-Civita connection and the Riemann curvature to the Finsler setting. In particular, the Finsler Ricci curvature generalizes the classical Ricci \emph{tensor}.

This implies that, the Cartan tensor and the Landsberg tensor measure the deviation of a Finsler spacetime to a pseudo-Riemannian one. In fact, a common classification of Finsler spacetimes (and Finsler spaces) is according to their ``distance'' to pseudo-Riemannian geometry in the following sense: 
\begin{align}\label{FinslerLadder}
    \textrm{pseudo-Riemannian}\subset \textrm{Berwald type}\subset\textrm{Landsberg type}\subset\textrm{weakly Landsberg}\subset\textrm{general Finsler}\,.
\end{align}
Berwald type spacetimes are Finsler spacetimes whose Cartan tensor is non-vanishing and whose Cartan non-linear connection defines an affine connection on spacetime:
\begin{align}\label{eq:BerwGeom}
    C_{\mu\nu\rho}\neq 0,\quad N^\mu{}_\nu(x,\y)= \Gamma^\mu{}_{\nu\sigma}(x)\,\dot x^\sigma\,.
\end{align}
Landsberg Finsler spacetimes and weakly Landsberg ones are those Finsler spacetimes whose Landsberg tensor vanishes, $P_{\mu\nu\rho}=0$, or whose trace of the Landsberg tensor vanishes, $P_\mu=0$, resp. These classes are of interest since for them the Finsler gravity equation simplifies severely, as we will recall next, and the deviation from pseudo-Riemannian is, in principle, not arbitrary large. Landsberg spaces that are not of Berwald type have been found and studied intensively \cite{asanov_unicorns,Bao_unicorns,shen_unicorns,elgendi_2020,Elgendi2021b,Elgendi2021a,Heefer:2023a}. 

The geometric objects we just recalled allow us to construct field equations for Finsler spacetimes that determine the Finsler spacetime function, analogously as how the Einstein equations determine the Lorentzian spacetime metric in general relativity.

\subsection{Finsler gravity}\label{ssec:Fgrav}
To use Finsler spacetime geometry instead of the Riemannian one as an attempt to construct an improved description of gravity is an obvious step that has been and is investigated intensively \cite{Ikeda1981,ChangLi2009,Kouretsis:2008ha,Voicu2010,GarciaParrado2022}. Among the many approaches, one is distinguished: it takes the step from Riemannian to Finsler geometry seriously and uses the Finsler function $F$ as fundamental variable. It is action-based and could prove that it naturally couples matter, modeled as a kinetic gas, to the Finslerian geometry of spacetime \cite{HPV19,Pfeiferetal:2020}. Further, it has been demonstrated that this coupling of Finsler geometry to kinetic gases predicts an accelerated expansion of the universe without the need to introduce dark energy as an additional substance \cite{Pfeifer_2025}.

The starting point of this approach is the following. From now on, let us refer to the Ricci curvature $\mathcal{R}$ of \eqref{eq:def_curvatures} as the \emph{curvature scalar}; then, the standard Finsler gravity action is 
\begin{align}\label{eq:gravact}
   \mathscr{S} := \int_{\Sigma} d\Sigma\ \frac{\mathcal{R}}{F^2}\,.
\end{align}
Recall that $\Sigma$ is the unit-tangent, or indicatrix, bundle
\begin{align}
    \Sigma = \{(x,\y)\in\mathcal{T}\colon\;F(x,\dot x) = 1\}\,;
\end{align}
it has a volume form defined with the help of the Liouville vector field $\mathbb{C} = \y^\mu \vd_\mu$, as
\begin{align}
     d\Sigma = \frac{\det g^{(2)}}{F^2}\ i_\mathcal{{\mathbb{C}}}(d^4x \wedge d^4\dot x)\,.
\end{align}
The action \eqref{eq:gravact} can actually be interpreted in two ways:
\begin{enumerate}
    \item The purely metric approach: As an action for the Finsler function $F$ (or any of its powers $F^n$) alone, $\mathscr{S}=\mathscr{S}[F]$. This means that the curvature scalar is derived from the \emph{canonical} non-linear connection.
    \item The Palatini approach: As an action for the Finsler function $F$ and an independent non-linear connection $\bar N$, namely $\mathscr{S}=\mathscr{S}[F,\bar N]$. This means that, while the volume form is defined in terms of $F$, the curvature scalar is constructed from the \emph{independent} non-linear connection coefficients $\bar N^\mu{}_\nu$. We shall write this scalar as $\mathcal{\bar R}$ (its construction being formally identical to \eqref{eq:def_curvatures}).
\end{enumerate}
In both cases, variation of the action with respect to the independent fields yields equations that determine the Finsler geometry of spacetime dynamically. All details on the derivation and the matter coupling can be found in the articles~\cite{HPV19,Hohmann:2021zbt,villasenor:2022}, while here we focus on the vacuum field equations.
\begin{enumerate}
    \item The purely metric field equation:
    \begin{align}\label{eq:FinsGravL}
        g^{(2)\,\mu\nu}\vd_{\mu}\vd_{\nu}\mathcal{R}-\frac{6}{F^2}\mathcal{R}+2g^{(2)\,\mu\nu}\left(2 \vd_\mu\nabla P_\nu-\nabla(\vd_\mu P_\nu)\right)=0.
    \end{align}
    \item The Palatini field equations:
    \begin{align}
        0 &= g^{(2)\,\mu\nu}\vd_{\mu}\vd_{\nu}\mathcal{\bar R}-\frac{6}{F^2}\mathcal{\bar R} \,,\label{eq:PalatR}\\
        0 &=\left(-2 P_\mu + (6 \frac{\dot x_\nu}{F^2} - 2 C_\nu) J^\nu{}_{\mu} - (\dot\partial_\nu J^\nu{}_\mu+\dot\partial_\mu J^\nu{}_\nu)\right)(\delta^{\mu}{}_\lambda \dot x^\tau - \delta^{\tau}{}_\lambda \dot x^\mu)
        - (\vd_{\rho}J^\tau{}_\lambda - \vd_{\lambda}J^\tau{}_\rho)\,\dot x^\rho \,, \label{eq:PalatN}
    \end{align}
    where $J^\mu{}_\nu = \bar N^\mu{}_{\nu} - N^\mu{}_{\nu}$ is the difference between the connection coefficients of the independent non-linear connection and the canonical Cartan ones.
\end{enumerate}
Both types of field equations possess a part which contains only the curvature tensors. In the metric approach, this part is relevant, for example, for Finsler spacetimes for which $g^{(2)\,\mu\nu}\left(2 \vd_\mu\nabla P_\nu-\nabla(\vd_\mu P_\nu)\right)=0$; in particular, for Berwald, Landsberg and weakly Landsberg geometries. In the Palatini approach, \eqref{eq:PalatR} needs always to be solved.

In the following we will show under which conditions the equation
\begin{align}
    g^{(2)\,\mu\nu}\vd_{\mu}\vd_{\nu}\mathcal{\bar R}-\frac{6}{F^2}\mathcal{\bar R} &= 0
\end{align}
reduces to 
\begin{align}
    \mathcal{\bar R} = 0\,,
\end{align}
for both the metric ($\mathcal{\bar R} = \mathcal{ R}$) and the Palatini cases.

%%%%%%%%%%%%%%%%%%%%%%%%%%%%%%%%%%%%%%%%%%%%%%%%%%%%%%%%%%%%%%%%%%%%%%%%%%%%%%%%
%%%%%%%%%%%%%%%%%%%%%%%%%%%%%%%%%%%%%%%%%%%%%%%%%%%%%%%%%%%%%%%%%%%%%%%%%%%%%%%%
\section{\texorpdfstring{$F^n$}{Fn} instead of \texorpdfstring{$F^2$}{F2}}\label{sec:F^n}
The conventional approach to Finsler geometry and Finsler gravity, i.e. to derive the geometric objects from $L=F^2$ as discussed in Section \ref{ssec:geomobj}, distinguishes the second power of the Finsler function $F$ among all other powers. The (mostly historical) reason for this is that it guarantees that:
\begin{itemize}
    \item $L$ is smooth in pseudo-Riemannian geometry, for which  $L =g_{\mu\nu}(x)\,\dot x^\mu \dot x^\nu$.
    \item The Finsler metric $g^{(2)}$ is $0$-homogeneous.
\end{itemize}
However, when one considers general Finsler spacetime geometry, the first property, that $L$ is smooth \emph{everywhere}, generically does not hold anymore.

In the following, we do not want to single out any specific power of $F$, but rather consider arbitrary powers $F^n$, an idea that appeared already in the articles \cite{Pfeifer:2011xi,Pfeifer:2011tk} but was not explored in subsequent works. In particular, we will demonstrate that when there exists some power $n>1$ of $F$ such that $F^n$ is smooth and the $\dot x$-Hessian of $F^n$ is non-degenerate, the Finsler gravity vacuum equations \eqref{eq:FinsGravL} and \eqref{eq:PalatR} simplify considerably. 

This means that the choice of working with $F^2$ instead of any other power of $F$ is  nothing more than a choice. The information about the geometry is contained in $F$ and one can work with any power of $F$, as it is convenient for the problem at hand.

To understand the Finsler gravity equation in terms of $F^n$, we convert the usual Finsler geometric objects derived from $F^2$ into how they can be expressed in terms of $F^n$.

\subsection{Finsler geometry from \texorpdfstring{$F^n$}{dd}}
For the rest of the article, $n\in\mathbb{R}$ will be a real number different from $0$ and $1$. Several of the relations between the $F^2$ and $F^n$ objects have been derived in \cite{Pfeifer:2011tk,Pfeifer:2011xi} in a slightly different language. We adapt the results here to the language of the present paper.

We begin by introducing the $n$-metric components, which are homogeneous of degree $n-2$ with respect to $\dot x$, by writing
\begin{align}
    g^{(n)}_{\quad \mu\nu}:=\frac{1}{2}\left(F^n\right)_{\cdot\mu\cdot\nu}\,.
\end{align}
It is related to the usual Finsler metric components and their inverse through 
\begin{align}
    g^{(2)}_{\quad \mu\nu}&=\frac{2}{n}F^{2-n}g^{(n)}_{\quad \mu\nu}+\left(2-n\right)F_{\cdot \mu} F_{\cdot \nu}\,,\label{rel metrics}\\
    g^{(2)\,\mu\nu}&=\frac{n}{2}F^{n-2}g^{(n)\,\mu\nu}+\frac{n-2}{n-1}F^{-2}\,\y^\mu\y^\nu\,,
    \label{rel inverse metrics}\\
    \det g^{(2)}{}_{\mu\nu}&=\frac{16}{n^4\left(n-1\right)} F^{8-4n}\det g^{(n)}{}_{\mu\nu}
    \label{rel determinants}\,.
\end{align}
These identities imply the following relation (that is useful to rewrite the Finsler gravity in \eqref{eq:u2}) expressing $g^{(2)\,\mu\nu}$ in $\delta^\mu_\sigma = g^{(2)\,\mu\nu}g^{(2)}_{\quad \sigma\nu}$ through $g^{(n)\,\mu\nu}$:
\begin{align}\label{eq:gng2}
     g^{(n)\,\mu\nu}g^{(2)}_{\quad \sigma\nu} = \left(\delta^\mu_\sigma - \frac{n-2}{n-1}\frac{\dot x^\mu \dot x_\sigma}{F^2} \right) \frac{2}{n}F^{2-n}\,.
\end{align}
For Finsler spacetime functions, we have that $g^{(2)}$ is non-degenerate in $\mathcal{T}$ (and therefore $g^{(n)}$ is non-degenerate in the timelike bundle too). But when one wants to study the geometry beyond $\mathcal{T}$, extending $F$ by zero to its boundary, there exist many examples, such as $m$-th rooth metrics with $m>2$
\begin{align}\label{eq:mth}
    F = (G_{\mu_1\ldots \mu_m}(x)\,\dot x^{\mu_1}\ldots\dot x^{\mu_m})^{\frac{1}{m}}\,,
\end{align}
for which $g^{(2)}$ does not exist, say, where $F=0$. However, note that for these, $g^{(m)}$ is perfectly well defined at the null directions.

Further, we define the $n$-Cartan tensor as
\begin{align}
    C^{(n)}{}_{\mu\nu\rho} := \frac{1}{2} g^{(n)}_{\quad \mu\nu \cdot\rho}
    =\frac{1}{4}\left(F^n\right)_{\cdot\mu\cdot\nu\cdot \rho}\,.
\end{align}
As for the metric, we can express the conventional Cartan tensor $C_{\mu\nu\rho} = C^{(2)}{}_{\mu\nu\rho}$ as
\begin{align}
    C^{(2)}{}_{\mu\nu\rho} 
    = \frac{2}{n} F^{2-n}C^{(n)}{}_{\mu\nu\rho}
    + (2-n)(F_{\cdot \mu} F_{\cdot \nu \cdot \rho} + F_{\cdot \nu} F_{\cdot \mu \cdot \rho} + F_{\cdot\rho} F_{\cdot \mu \cdot \nu}) + (n-1)(2-n) F^{-1} F_{\cdot \mu} F_{\cdot \nu} F_{\cdot \rho}\,.
\end{align}

For the geodesic spray, remarkably, no distinguished power $F^n$ (with $n\neq 0,1$) exists. It can equally well be derived from any power, as the following holds:
\begin{align}
    G^\mu 
    &= \frac{1}{4}g^{(2)}{}^{\mu\nu}\left(\dot x^\sigma \partial_\sigma \dot \partial_\nu F^2 - \partial_\nu F^2\right)\\
    &= \frac{1}{4}g^{(n)}{}^{\mu\nu}\left(\dot x^\sigma \partial_\sigma \dot \partial_\nu F^n - \partial_\nu F^n\right) \label{eq:spray_Fn}\\
    &=: G^{(n)\mu}\,.
\end{align}
The proof of this equality just uses Euler's Theorem for homogeneous functions and the relation \eqref{rel inverse metrics}, see for example~\cite{Pfeifer:2011tk,Pfeifer:2011xi}. This invariance property immediately propagates to the geometric objects that are derived solely from the derivatives of the geodesic spray, namely the canonical non-linear connection, the non-linear curvature tensor and the curvature scalar. Moreover, it can be shown that the Landsberg tensor with one index raised also is independent of the power $n$ of the Finsler function that one employs for its derivation:
\begin{align}
    N^\mu{}_{\nu} = N^{(n)\mu}{}_{\nu},\quad R^\mu{}_{\nu\rho} = R^{(n)\mu}{}_{\nu\rho}, \quad \mathcal{R} = \mathcal{R}^{(n)},\quad P^{\mu}{}_{\nu\rho} = P^{(n)\mu}{}_{\nu\rho}\,.
    \label{eq:curvatures_Fn}
\end{align}

This remarkable observation, that many geometric objects in Finsler geometry, and thus Finsler gravity, are independent of whether they are derived from $F^2$ or $F^n$ allows us to rewrite the Finsler gravity equation easily in terms of the latter.

\subsection{Finsler gravity in terms of \texorpdfstring{$F^n$}{dd}}
The canonical action-based Finsler gravity equation \eqref{eq:FinsGravL} can be compactly rewritten as
\begin{align}
    0
    & =  g^{(2)\,\mu\nu}\vd_{\mu}\vd_{\nu}\mathcal{R}-\frac{6}{F^2}\mathcal{R}+2g^{(2)\,\mu\nu}\mathcal{P}_{\mu\nu}\,
    \label{eq:FGravVac_0}
\end{align}
when we define
\begin{align}
   \mathcal{P}_{\mu\nu}
     & :=  2 \vd_\mu\nabla P_\nu-\nabla(\vd_\mu P_\nu).
\end{align}
The terms $\mathcal{R}$ and $\mathcal{\mathcal{P}_{\mu\nu}}$ do not change when changing the variable from $F^2$ to $F^n$. Expressing the appearing $g^{(2)}$ in terms of $g^{(n)}$ yields
\begin{align}\label{eq:FGravVac}
    0 = \left(\frac{n}{2}F^{n-2}g^{(n)\,\mu\nu}+\frac{n-2}{n-1}F^{-2}\y^\mu\y^\nu\right)\vd_{\mu}\vd_{\nu}\mathcal{R} 
    - \frac{6}{F^2}\mathcal{R}
    +2\left(\frac{n}{2}F^{n-2}g^{(n)\,\mu\nu}+\frac{n-2}{n-1}F^{-2}\y^\mu\y^\nu\right)\mathcal{P}_{\mu\nu}\,.
\end{align}
Using that $\mathcal{R}$ is $2$-homogeneous with respect to $\dot x$ and that $\dot x^\mu \dot x^\mu \mathcal{P}_{\mu\nu} = 0$ allows us to rewrite the equations as
\begin{align}\label{eq:FullFinslern}
    0 = \frac{n}{2}g^{(n)\,\mu\nu}\vd_{\mu}\vd_{\nu}\mathcal{R} 
    + \frac{2(1-2n)}{n-1} \frac{\mathcal{R}}{F^n} 
    + ng^{(n)\,\mu\nu}\mathcal{P}_{\mu\nu}\,.
\end{align}

Similarly, the Finsler Palatini equation \eqref{eq:PalatR} can be expressed  as
\begin{align}\label{eq:Fgrav2}
    0 = \frac{n}{2}g^{(n)\,\mu\nu}\vd_{\mu}\vd_{\nu}\mathcal{\bar R} 
    + \frac{2(1-2n)}{n-1} \frac{\mathcal{\bar  R}}{F^n}\,.
\end{align}

The question to be answered next is: Under which conditions do these equations reduce to $\mathcal{R} = 0$ or $\mathcal{\bar R} = 0$?

%%%%%%%%%%%%%%%%%%%%%%%%%%%%%%%%%%%%%%%%%%%%%%%%%%%%%%%%%%%%%%%%%%%%%%%%%%%%%%%%
%%%%%%%%%%%%%%%%%%%%%%%%%%%%%%%%%%%%%%%%%%%%%%%%%%%%%%%%%%%%%%%%%%%%%%%%%%%%%%%%
\section{Reducing the vacuum field equations}\label{sec:RedVacEq}
We consider Finsler gravity in either the metric or the Palatini formalism. \vspace{10pt}

In the metric formalism, we restrict ourselves to Finsler spacetimes with $g^{(n)\mu\nu}\mathcal{P}_{\mu\nu} = 0$ for some $n$ (which then implies it for any $n$), a condition satisfied, for example, by Landsberg, weakly Landsberg and Berwald spaces, as for them $P_\mu = 0$, see the discussion at the end of Section \ref{ssec:geomobj}. Then, the metric Finsler gravity equation \eqref{eq:FullFinslern} reduces to
\begin{align}\label{eq:Fgrav1}
    0 = \frac{n}{2}g^{(n)\,\mu\nu}\vd_{\mu}\vd_{\nu}\mathcal{R} 
    + \frac{2(1-2n)}{n-1} \frac{\mathcal{R}}{F^n}\,.
\end{align}
In the Palatini formalism, one of the field equations has a similar form, namely \eqref{eq:Fgrav2} needs to hold.

Generically, $F$ and $F^2$ are not smooth on the light cones, i.e., on the boundary $\partial\mathcal{T}$ of the timelike cones $\mathcal{T}$, so neither the Finsler metric nor other derivatives are guaranteed to exist there. Yet, there are solutions of the Finsler gravity equations with a well defined geometry along the lightlike directions (as an example will show later) because the geodesic spray, the non-linear connection and the curvature are well defined there in terms of $F^n$.

We want to prove next that if one of the following lightlike regularity conditions on $F$ is satisfied, to be discussed in Section \ref{ssec:lightlike regularity}, the equations \eqref{eq:Fgrav1} and  \eqref{eq:Fgrav2} simplify to $\mathcal{ R} = 0$ and $\mathcal{\bar R} = 0$, respectively.

To do so, we first state conditions on $F^n$ that have to hold on the light cones, in order to then prove the statement with two different techniques in Sections \ref{ssec:maximumprinciple} and \ref{ssec:divbyzero}.

\subsection{Lightlike regularity conditions}\label{ssec:lightlike regularity}
Among the literature on Finsler spacetimes, interestingly, there is no consensus on how regular $F$ or $F^2$, or any power of $F$, needs to be on $\partial\mathcal{T}$, the set of lighlike directions, where $F(x,\dot x) = 0$. Depending on the application in mind, different authors require different degrees of differentiability for $F$ on the light cones and, by doing so, exclude or include certain examples from their definitions, e.g. compare \cite[Defs. 3.1 and 3.5]{Javaloyes:2018lex} with \cite[App. A]{Javaloyes:2018lex} or \cite{Pfeifer:2011xi,Beem,minguzzi_2016,Lammerzahl:2012,Hohmann:2021zbt}. The different conditions that one may impose on (the extension of) some power of $F$ at parts of $\partial \mathcal{T}$, we refer to as \emph{lightlike regularity conditions}. 

For our goal, to ensure that the Finsler vacuum equations of the type \eqref{eq:Fgrav1} and \eqref{eq:Fgrav2} can be simplified (leading to generalizations of the results from \cite{villasenor:2022}), one of the following lightlike regularity conditions needs to hold:
\begin{itemize}
    \item[{\bf (a)}] There exists a real number $n\geq2$ such that $F^n$ is at least eight times differentiable ($C^8$) and $g^{(n)}{}_{\mu\nu}$ is non-degenerate.
    \item[{\bf (b)}] There exist a positive number $n\neq 1$ such that $\frac{3+\sqrt{
    9-\frac{16}{n}+\frac{8}{n^2}}}{2}\notin\mathbb{N}$ and some subset $B\subseteq\partial \mathcal{T}$ such that the function $F^n$ is smooth on $\mathcal{T}\cup B$ and $\left.g^{(n)}_{\quad \mu\nu}\right|_{B}$ is non-degenerate. The rather technical square root condition includes an easily memorizable case, namely: $n>1$ is sufficient already.
\end{itemize}

The first condition ensures that $F^n$ is as often differentiable and that $g^{(n)}$ and its inverse exist on the light cones as it is needed to execute the computations for the proof of the reduction of the field equations. This is two times more differetiable as it would be necessary to just write down the field equations. 

The second condition is surprising, in that it basically requires that $F^n$ is smooth ($C^\infty$) and $g^{(n)}$ and its inverse exist \emph{just} on a subset $B$ of each light cone that can be as small as just one lightlike direction. 

Comparing the two conditions reveals that the first requires less regularity ($C^8$ instead of $C^\infty$) but on a larger set (the whole of $\partial \mathcal{T}$ instead of only on $B$), while the second requires stronger regularity, but on a smaller set. If either {\bf (a)} (Section \ref{ssec:maximumprinciple}) or {\bf (b)} (Section \ref{ssec:divbyzero}) hold, we will prove that the Finsler gravity equation \eqref{eq:Fgrav1} reduces to $\mathcal{R} = 0$ and \eqref{eq:Fgrav2} reduces to $\mathcal{\bar R} = 0$.

As the conditions {\bf (a)} and {\bf (b)} are rather technically precise, one can retain the following simple case. The reduction of the field equations is true if there exists an $n \geq 2$ with $F^n$ smooth and $g^{(n)}$ non-degenerate on the light cones $\partial\mathcal{T}$.

\subsection{Reducing the vacuum equations from the elliptic maximum principle} \label{ssec:maximumprinciple}
In \cite[Th. 5.18]{villasenor:2022} it was proven that the equations
\begin{align}\label{eq:Fgrav3}
    0 = g^{(2)\,\mu\nu}\vd_{\mu}\vd_{\nu}\mathcal{ R} 
    -6 \frac{\mathcal{  R}}{F^2}\,,\quad
    0 = g^{(2)\,\mu\nu}\vd_{\mu}\vd_{\nu}\mathcal{\bar R} 
    -6 \frac{\mathcal{\bar  R}}{F^2}\,
\end{align}
are equivalent to $\mathcal{  R}=0$ or $\mathcal{\bar  R}=0$ under some assumptions on $F^2$.  The strong conditions imposed to derive the result were that $F^2$ is more regular than just continuous on $\partial \mathcal{T}$, where $F=0$. Moreover, it was required that $g^{(2)}{}_{\mu\nu}$ is non-degenerate there. Very often, this is not the case for many examples of Finsler spacetimes, as for example the $m$-th root metrics demonstrate \eqref{eq:mth}.

We saw in \eqref{eq:Fgrav1} and \eqref{eq:Fgrav2} that one can reformulate equations of the type \eqref{eq:Fgrav3} in terms of $F^n$ and $g^{(n)}{}_{\mu\nu}$, which for many examples of Finsler spacetimes have improved regularity and non-degeneracy properties. Hence, we seek to generalize the statement to the case in which perhaps $F^2$ does not have good regularity properties at the directions where $F=0$ but $F^n$ does for \emph{some} choice of $n$. We will achieve this in the Theorems \ref{thm:maximumprinciplemain} and \ref{thm:main extension 1} as applications of the classical\footnote{We always refer to the maximum principle proved by Eberhard Hopf in 1927 \cite{Hopf1927}, of which we provide a statement for the reader's convenience. Its setting is that of a partial differential equation (PDE, from now on) $a^{ij}(z)\,\partial^2_{z^i z^j} u+b^k(z)\,\partial_{z^k}u+c(z)\,u=0$, where the scalar unknown $u(z^1,\ldots,z^m)$ is of class $C^2$ on an open domain $\Omega\subseteq\mathbb{R}^m$. Locally in $\Omega$, the symmetric matrix $a^{ij}$ must be \emph{uniformly elliptic}: denoting by $\lambda(z)$ and $\Lambda(z)$, respectively, its smallest and largest eigenvalues, $\lambda>0$ everywhere and $\Lambda/\lambda\leq const$. One also requires $\left|b^k\right|/\lambda$ and $c/\lambda$ to be bounded (but, again, only locally). Then, if $c\leq 0$ and $u$ attains a non-negative absolute maximum in $\Omega$, it is in fact the case that $u$ is constant. \label{foot:max_prin_summary}} strong maximum principle \cite[Th. 3.5]{gilbarg-trudinger}. We discuss all the details for the metric Finsler case first, \eqref{eq:Fgrav1}, before we discuss how the results carry over to the Palatini case \eqref{eq:Fgrav2}.

The key insight from \cite{villasenor:2024} is to define the variable  
\begin{equation}
    u:=\frac{\mathcal{R}}{F^2},
    \label{u}
\end{equation}
which is $0$-homogeneous, and to rewrite the equations of interest as an equation for $u$. 

\subsubsection{The metric Finsler gravity equation}
Rewriting \eqref{eq:Fgrav1} in terms of $u$, we find that for any $n$, except $0$ and $1$, the metric gravity equation can be cast as follows, using for example \eqref{eq:gng2},
\begin{align}
    0 
    &= \frac{n}{2}F^{n-2} g^{(n)\,\mu\nu}\vd_{\mu}\vd_{\nu}\mathcal{R} 
    + \frac{2(1-2n)}{n-1} u \label{eq:u}\\
    &= \frac{n}{2}F^{n} g^{(n)\,\mu\nu}\vd_{\mu}\vd_{\nu}u + 2 u\,,\label{eq:u2}
\end{align}
which we interpret as a partial differential equation determining the $\dot x$-dependence of $u$. The interesting observation here is that the first term in \eqref{eq:u} depends on $F^n$, while in the second, $u$ is defined from $F^2$. This allows us to draw conclusions about the regularity of $u$ from the regularity properties of $F^n$.

We stress that on $\mathcal{T}$, the above equations \eqref{eq:u} and \eqref{eq:u2} are equivalent representations of the same equation for any $n\in\mathbb{R}\setminus\left\{0,1\right\}$, and we will use specific choices of $n$ to draw our conclusions.

To make this precise, we formulate the following mathematical lemma:
\begin{lemma} 
\label{continuous extension}
    Assume that $F$ solves the metric Finslerian vacuum equation on $\mathcal{T}$ with $g^{(n)\mu\nu}\mathcal{P}_{\mu\nu}=0$, which means that \eqref{eq:u} holds on the set of timelike vectors. Then, we can make the following statements.
    \begin{enumerate}[label=\roman*)]
        \item \label{continuous extension 1} Assume that there exists a constant $n\in\mathbb{R}$ with $n\geq 2$ such that the function $F^n$ is at least six times differentiable $(C^6)$ on $\overline{\mathcal{T}}=\mathcal{T} \cup \partial\mathcal{T}\subseteq TM\setminus\mathbf{0}$ 
        and $g^{(n)}_{\quad\mu\nu}$ is non-degenerate on all of the light cones $\partial\mathcal{T}$. Then, $u$ in \eqref{u}, \eqref{eq:u} and \eqref{eq:u2} must be continuous on the timelike and null directions $\overline{\mathcal{T}_x}  \subseteq T_xM\setminus \left\{0\right\}$:
        \begin{align}
            u\in C^0(\overline{\mathcal{T}_x})\,.  
        \end{align}
        Consequently, $u$ attains its absolute maximum and absolute minimum on $\overline{\mathcal{T}_x}$.
        
        \item \label{continuous extension 2} Assume that there exists a constant $n\in\mathbb{R}$ with $n\geq 2$ for all $x\in M$ such that $F^n$ is eight times differentiable $(C^8)$ and $g^{(n)}_{\quad\mu\nu}$ is non-degenerate (i.e., that condition {\bf (a)} in Section \ref{ssec:lightlike regularity} holds). Then, $u$ is continuous and must be $0$ on all the light cone directions:
        \begin{equation}
            \left.u\right|_{\partial(\mathcal{T}_x)}=0\,.
            \label{u at boundary}
        \end{equation}
    \end{enumerate}
    The properties obtained hold for $u$ at each individual $x\in M$.
\end{lemma}

\begin{proof}\,
    \begin{enumerate} [label=\roman*)]
        \item On the right-hand side of \eqref{eq:u}, the term $g^{(n)\,\mu\nu}$ extends continuously to $\overline{\mathcal{T}_x}\subseteq\overline{\mathcal{T}}\cap T_xM$ by the assumption that $F^n$ is $C^6$ and $g^{(n)}{}_{\mu\nu}$ is non-degenerate there. The same assumptions guarantee the existence and continuity of $\dot\partial_\mu\dot\partial_\nu\mathcal{R}$, as this term contains at most 6th derivatives of $F^n$.

        Moreover, since we assumed that $F^n$ is $C^6$, also $F^{n-2}$ is at least continuous on $\overline{\mathcal{T}_x}$, as we imposed that $n-2\geq 0$. Thus, from our assumptions it follows that the complete term 
        \begin{align}\label{eq:prfL1}
            \frac{n}{2}F^{n-2} g^{(n)\,\mu\nu}\vd_{\mu}\vd_{\nu}\mathcal{R}
        \end{align}
        is continuous on $\overline{\mathcal{T}_x}$, and as it is equated to $u$ in \eqref{eq:u}, so is $u$.
        
        Having ensured that $u$ is continuous, and by its $0$-homogeneity w.r.t. $\dot x$, one can apply Proposition \ref{properties proj} of the Appendix \ref{app:TechDet} and thus $u$ attains its absolute maximum and absolute minimum among the directions in $\overline{\mathcal{T}_x}$.

        \item Here, we distinguish two cases, the first when $n>2$, and the second when $n=2$. 
        
        For $n>2$, the assumptions that $F^n$ is $C^8$ and that $g^{(n)}{}_{\mu\nu}$ is non-degenerate not only guarantee that the term
        \begin{align}
            \frac{n}{2}F^{n-2} g^{(n)\,\mu\nu}\vd_{\mu}\vd_{\nu}\mathcal{R}
        \end{align}
        is well defined on $\partial(\mathcal{T}_x)$, but also that it vanishes there, as $F^{n-2}=0$ on $\partial(\mathcal{T}_x)$. Thus, by \eqref{eq:u}, also $u$ must vanish as claimed. For this part of the proof, it would be sufficient if $F^n$ is only $C^6$.

        For\footnote{This case is covered in the proof of \cite[Lem. 2.6]{villasenor:2024}, at the end of case (A).} $n=2$, we find from \eqref{eq:u} that
        \begin{align}
             u=\frac{1}{6} g^{(2)\,\mu\nu}\vd_{\mu}\vd_{\nu}\mathcal{R}\,.
        \end{align}
        Since we now assumed that $F^2$ is $C^8$ on $\partial\mathcal{T}$, it follows that $u$  is $C^2$ on $\partial(\mathcal{T}_x)$ and thus the equation \eqref{eq:u2}, namely
        \begin{align}
            0 = F^2 g^{(2)\,\mu\nu}\vd_{\mu}\vd_{\nu}u + 2 u,
        \end{align}
        implies that $u = 0$ on $\partial(\mathcal{T}_x)$.
    \end{enumerate}
    These properties hold at each $x\in M$, which is a consequence of the assumption that the distinguished $n$ is the same constant for all $x\in M$.
\end{proof}
Thus we established that if $F^n$ is regular enough (at least $C^8$, but in particular when $F^n$ is smooth) on the light cones, and the corresponding $n$-metric is non-degenerate, then the variable $u$ must vanish on the light cones for solutions of the Finsler gravity equation with $g^{(n)\mu\nu}\mathcal{P}_{\mu\nu}=0$. More precisely, we can conclude that
\begin{align}
   \mathcal{R} = F^2 {u}\,
\end{align}
with $u$ and $F^2$ being at least continuous\footnote{ In Section \ref{ssec:divbyzero}, where we will use the regularity assumption~{\textbf {(b)}} of \ref{ssec:lightlike regularity}, in Lemma \ref{recursion lemma} we will obtain an analogous result. However, there, due to the different regularity assumptions, the statement allows us to even conclude that $\mathcal{R}$ can be decomposed into smooth factors.} on each light cone. 

These insights about the behavior of solutions of the Finsler gravity equation on light cones allow us to conclude that, under these regularity conditions, the equation reduces to $\mathcal{R}=0$ on the timelike and lightlike directions $\overline{\mathcal{T}}$. 

\begin{theorem} \label{thm:maximumprinciplemain}
Let $F$ be such that there exists a constant $n\geq2$ for all $ x\in M$ for which $F^n$ is at least $C^8$ and  $g^{(n)}{}_{\mu\nu}$ is non-degenerate on the light cones (lightlike regularity assumption {\textbf {(a)}} in Section \ref{ssec:lightlike regularity}). Moreover, assume that $g^{(n)\mu\nu}\mathcal{P}_{\mu\nu}=0$. Then, the metric Finsler gravity vacuum equation 
\begin{align}
        0 = g^{(2)\,\mu\nu}\vd_{\mu}\vd_{\nu}\mathcal{ R}-6 \frac{\mathcal{  R}}{F^2} + 2 g^{(2)\mu\nu}\mathcal{P}_{\mu\nu}
\end{align}
reduces to
\begin{align}
    \mathcal{ R} = 0\,.
\end{align}
\end{theorem}

To prove this theorem, we follow the steps outlined in \cite[Lem. 2.6 and Th. 2.6]{villasenor:2024}:
\begin{itemize}
    \item Introduce coordinates
    \begin{align}
        r(\y):=F(x,\y), \qquad z^i(\y):=z^i_\Sigma (\tfrac{\y}{F(x,\y)}),\ i=1,2,3\,,
    \end{align}
   on each tangent space $T_xM$, where $\left(z^i_\Sigma\right)$ is a pre-chosen (global, $C^\infty$) coordinate system on the manifold $\Sigma_x = \{\dot x\in \mathcal{T} \colon F(x,\dot x) = 1\}$.
    \item In these coordinates, any $m$-homogeneous function $h$ on $\mathcal{T}$ can be written as
    \begin{align}
        h(x,\dot x) = F(x,\dot x)^m h(x,\tfrac{\dot x}{F(x,\dot x)}) = r^m h(x,z) = r^m h_\Sigma\,.
    \end{align}
    \item In $T_xM$, for our zero-homogeneous variable of interest $u$ (recall \eqref{u}), this means that
    \begin{align}
        u = u_\Sigma(z)\,.
    \end{align}
    \item The Finsler gravity equation \eqref{eq:u2} can be rewritten in its $n=2$ representation as an elliptic PDE for $u_\Sigma(z)$ see \cite[(2.58) to (2.59) and (A) of Lem. 2.6.]{villasenor:2024}:
    \begin{align}\label{eq:uelipt}
        a^{ij}\partial^2_{z^i z^j} u_\Sigma+b^k\partial_{z^k}u_\Sigma+c\,u_\Sigma=0\,,
    \end{align}
    with 
    \begin{align}
        a^{ij} &= -r^2\left(g^{(2)}\right)^{-1}(\Dd z^i,\Dd z^j) = -r^2g^{(2)\,\mu\nu}(z)\,\frac{\partial z^i}{\partial\y^\mu}\frac{\partial z^j}{\partial\y^\nu}\,,\\
        b^k &= -r^2g^{(2)\,\mu\nu}\frac{\partial^2 z^k}{\partial\y^\mu\partial\y^\nu}\,,\\
        c &=-2\,.
    \end{align}
    \item Applying the strong maximum principle (see footnote \ref{foot:max_prin_summary}) to \eqref{eq:uelipt} proves that $u=0$ and thus $\mathcal{R}=0$.
\end{itemize}
Let us detail this last step more formally:
\begin{proof}
The assumptions of Theorem \ref{thm:maximumprinciplemain} immediately reduce the Finsler gravity vacuum equation to \eqref{eq:u2} with $n=2$.

Assuming the regularity condition {\bf (a)} of Section \ref{ssec:lightlike regularity} means that both statements of Lemma \ref{continuous extension} hold. Part i) guarantees that $u$ assumes its absolute maximum and its minimum at some directions in $\overline{\mathcal{T}_x}$. Denote:
\begin{itemize}
    \item by $\dot x_{\alpha +} \in T_xM$, the directions for which $u(\dot x_{\alpha +}) =  \underset{_{\overline{\mathcal{T}_x}}}{\max }\,u $;
    \item by $\dot x_{\alpha -} \in T_xM$, the directions for which $u(\dot x_{\alpha -}) = \underset{_{\overline{\mathcal{T}_x}}}{\min }\,u$.
\end{itemize}
Here, $\alpha\in\mathbb{R}$: in principle, there can be infinitely many of these directions, even as many as real numbers. Let us focus on the maxima $\y_{\alpha +}$. There can be either some $\y_{\alpha +}\in \mathcal{T}_x$, or all of the $\y_{\alpha +}$ are in $\partial(\mathcal{T}_x)$. 

Assume there is at least one $\y_{\alpha +}$ in $ \mathcal{T}_x$. Then, we can normalize it to define $z_{\alpha +}:=\dot x_{\alpha +}/F(x,\dot x_{\alpha +})\in \Sigma_x$. Hence, $u_\Sigma$ attains its maximum there while solving the elliptic PDE \eqref{eq:uelipt}, and the strong maximum principle \cite[Th. 3.5]{gilbarg-trudinger} applies. This is the case because the coefficients in \eqref{eq:uelipt} satisfy $c=-2\leq 0$ and the appropriate local requirements, see footnote \ref{foot:max_prin_summary}.  (In turn, these requirements hold due to $a^{ij}$ and $b^k$ being continuous functions of just $(z^1,z^2,z^3)$, together with the positive definiteness of $(a^{ij})$.)

The following two cases exist:
\begin{itemize}
    \item[1.]  $u(\y_{\alpha +})$ is nonnegative, so the maximum principle applies straightforwardly.
    \item[2.]  $u(\y_{\alpha +})$ is negative. Then, there also exists a minimum $u(\dot x_{\alpha -})<0$, which by the statement \eqref{u at boundary} must also lie in $\mathcal{T}_x$, i.e., \ $\dot x_{\alpha -}\in \mathcal{T}_x$. Hence, it can be normalized to $z_{\alpha -}:=\dot x_{\alpha -}/F(x,\dot x_{\alpha -})\in \Sigma_x$ and the function $-u_\Sigma$ attains a positive maximum
    there, so the strong maximum principle can be applied to it.
\end{itemize}
In any case, the maximum principle implies that $\left.u\right|_{\Sigma_x}$ is constant. By $0$-homogeneity, $u$ must be constant in the whole tangent space $T_xM$, so $u(x,\y)=u(x)$. Even more, it is enough to use this information in \eqref{eq:u2} to find that actually $u=0$, and thus the desired reduction of the Finsler gravity equation follows:
\begin{align}
    \mathcal{R}= 0\,.
\end{align}
This reasoning can be carried out at all $x\in M$, as we assumed that the constant $n$ exists everywhere. 

To complete the proof in the missing cases, assume that there is no $\dot x_{\alpha +}\in \mathcal{T}_x$. Then, as long as there is at least one minimum $\dot x_{\alpha -}\in \mathcal{T}_x$, the whole argument can be repeated in terms of $-u(\dot x_{\alpha -})$. One concludes that either $\mathcal{R}= 0$ or no minimum lies in the set of timelike vectors at $x$.

Thus, only the case remains where all the maxima and minima lie at the lightlike directions, i.e. $\dot x_{\alpha +},\dot x_{\alpha -}\in \partial(\mathcal{T}_x)$
for all $\alpha$. Then, according to \eqref{u at boundary},
\begin{align}
    0 = u(\dot x_{\alpha -}) = \min u \leq u \leq \max u = u(\dot x_{\alpha +}) = 0\,:
\end{align}
also in this case, the conclusion is $u=0\Rightarrow \mathcal{R}= 0$.
\end{proof}

\subsubsection{The Palatini Finsler gravity equation} \label{palatini theory 1}
In order to conclude that the metric Finsler gravity equation \eqref{eq:Fgrav1} reduces to $\mathcal{R}=0$, we required certain regularity conditions on $F^n$. 

In the Palatini approach, see Section \ref{ssec:Fgrav}, there exist two independent variables, namely the independent non-linear connection coefficients $\bar N^\mu{}_\nu$ and the Finsler function $F$. Hence, to conclude that the Palatini Finsler gravity equation \eqref{eq:Fgrav2} reduces to $\bar{\mathcal{R}}=0$ with the same argument as in the metric case, we need to specify the corresponding conditions on $\bar N^\mu{}_\nu$ and $F$.

The original reduction \cite[Th. 2.7]{villasenor:2024} was obtained precisely for the Palatini theory, but under the rather strong assumption that $F^2$ and $g^{(2)}$ are regular on the light cones. To weaken the assumption, from $F^2$ to $F^n$, we rewrite the Palatini equation similarly as we did in the metric case:
\begin{align}
    0 
    &= g^{(2)\,\mu\nu}\vd_{\mu}\vd_{\nu}\mathcal{\bar R} -6 \frac{\mathcal{\bar  R}}{F^2}\\
    &= \frac{n}{2}F^{n-2} g^{(n)\,\mu\nu}\vd_{\mu}\vd_{\nu}\mathcal{\bar R} 
    + \frac{2(1-2n)}{n-1} u \label{eq:uPal}\\
    &= \frac{n}{2}F^{n} g^{(n)\,\mu\nu}\vd_{\mu}\vd_{\nu}u + 2 u\,,\label{eq:u2Pal}
\end{align}
where here $u$ is defined by the curvature scalar of the independent non-linear connection:
\begin{align}
    u := \frac{\bar{\mathcal{R}}}{F^2}\,.
\end{align}
Now, the arguments from the previous section about the strong maximum principle hold as long as $g^{(n)}_{\quad\mu\nu}$ and $\bar N^\mu{}_\nu$ extend to $\overline{\mathcal{T}}$ with enough regularity. Without repeating all the steps, we can state the result:

\begin{theorem} \label{thm:main extension 1}
    Let $(F,\bar N^\mu_{\nu})$ be the fields to be determined as a solution of the Finslerian Palatini equations~\eqref{eq:PalatR} and~\eqref{eq:PalatN} in 4 dimensions. 

    If $N^\mu_\nu$ is at least five times differentiable ($C^5$) on $\overline{\mathcal{T}}$, and if there exists an $n\in\mathbb{R}$ with $n\geq 2$ such thatrrr $F^n$ is at least $C^4$ and $g^{(n)}_{\mu\nu}$ is non-degenerate on $\overline{\mathcal{T}}$ (so in particular on $\partial \mathcal{T}$, where $F^n=0$), then the equation \eqref{eq:PalatR} reduces to
    \begin{equation}
        \bar{\mathcal{R}}=0.
        \label{reduced palatini eq.}
    \end{equation}
\end{theorem}

%%%%%%%%%%%%%%%%%%%%%%%%%%%%%%%%%%%%%%%%%%%%%%%%%%%%%%%%%%%%%%%%%%%%%%%%%%%%%%%%
%%%%%%%%%%%%%%%%%%%%%%%%%%%%%%%%%%%%%%%%%%%%%%%%%%%%%%%%%%%%%%%%%%%%%%%%%%%%%%%%
\subsection{Reducing the vacuum equations from a division by zero argument}\label{ssec:divbyzero}
In this section, we show that asking for regularity conditions of $F^n$ in all lightlike directions (lightlike regularity condition {\bf (a)} from Section \ref{ssec:lightlike regularity}), as we did in Theorem \ref{thm:maximumprinciplemain}, is not the only way to reduce the Finsler gravity equation. The alternative is to employ the lightlike regularity condition {\bf (b)} from Section \ref{ssec:lightlike regularity}, which requires stronger regularity but on a smaller set, that is, not all the lightlike directions. We refer to Appendix \ref{app:TechDet} for the necessary theory on functions on subsets of $TM$ and those on subsets of a fixed tangent space $T_xM$. The end results are stated in all detail in Theorems \ref{thm:analyticity main} and \ref{thm:analyticity_palatini}.

We leverage the technique originally introduced in \cite[\S 5.1]{villasenor:2022}. Abstractly, one may call it\footnote{In \cite[Lem. 5.4 and Th. 5.8]{villasenor:2022} it was used to prove uniqueness of $N^\mu_\nu$ in the Finslerian Palatini theory; in \cite[\S (C2)]{villasenor:2024} it was called \emph{division by powers of $F^2$ on $\partial \mathcal{T}$}. Here we sketch a general formulation, since it might be useful in other settings in the future.} \emph{infinite division by 0} and describe it as follows. 

Consider a function $v(\dot x)$ on $T_xM$ and assume that it solves a partial differential equation in which another smooth function $L(\dot x)$ appears. Moreover, assume that there exists some fixed $\dot x_0$ such that $L(\dot x_0)=0$. Then, we perform the following steps:
\begin{itemize}
    \item Try to use the PDE to recursively prove that the quotients $\frac{v}{L^m}$ with $m\in\mathbb{N}\cup\left\{0\right\}$ extend to $\dot x_0$ smoothly.

    \item \label{description 2} When successful, this implies that the derivatives of $v$ \emph{of all orders} vanish at $\dot x_0$.

    \item Assume that $v$ is analytic (in the real multivariable sense) including the point $\dot x_0$. From this,  one deduces that $v=0$ identically.
\end{itemize}

For the application of these steps to the Finsler gravity equations \eqref{eq:Fgrav1} or \eqref{eq:Fgrav2}, we fix $x\in M$ and use $\left.\mathcal{R}\right|_{\mathcal{T}_x}$ as the unknown $v$ (compare with Section \ref{ssec:maximumprinciple}). We start by investigating the case of the metric Finsler gravity equation and afterwards we adapt the arguments to the Palatini formulation.

\subsubsection{The metric Finsler gravity equation}
We first establish the following lemma; it leads  to the regularity condition \textbf{(b)} in Section \ref{ssec:lightlike regularity}.
\begin{lemma} \label{recursion lemma}
    Assume that $F$ solves the metric Finslerian vacuum equation in $\mathcal{T}$ with $g^{(n)\mu\nu}\mathcal{P}_{\mu\nu} =0$, which means that \eqref{eq:Fgrav1} holds. Suppose that for some $n\in\mathbb{R}\setminus\left\{0,1\right\}$ and $B\subseteq\partial \mathcal{T}$, the function $F^n$ extends smoothly to $\mathcal{T}\cup B$ with $\left.g^{(n)}_{\quad\mu\nu}\right|_B$ non-degenerate. \\
    If $\left(3m-m^2\right)n^2-4n+2\neq 0$ for all $m\in\mathbb{N}\cup\left\{0\right\}$, then, for each such $m$, the quotient $\left.\frac{\mathcal{R}}{F^{mn}}\right|_{\mathcal{T}_x}$ extends smoothly to $\left(\mathcal{T}\cup B\right)_x$; in other words, 
    \begin{equation}
        \mathcal{R}=\left(F^n\right)^m \mathcal{Q}_m \quad\hbox{\rm for some}\quad \mathcal{Q}_m\in C^\infty(\left(\mathcal{T}\cup B\right)_x).
        \label{thesis recursion lemma}
    \end{equation}
\end{lemma}
\begin{proof}
    By induction on $m$. In the case $m=0$, we only have to check that $\mathcal{R}$($=\mathcal{Q}_0$) extends smoothly $(C^\infty)$ to $\left(\mathcal{T}\cup B\right)_x  =\mathcal{T}_x\cup B_x$. But the assumed smoothness of $F^n$ and non-degeneracy of $g^{(n)\,\mu\nu}$ on $\left(\mathcal{T}\cup B\right)_x$ already guarantee this.
    
    Assume that the assertion holds for $m\in\mathbb{N}\cup\left\{0\right\}$ and let us prove it for $m+1$. \vspace{10pt}
    
    Abbreviate $L:=F^n$. Then, by assumption, we have that $\mathcal{R}=L^m \mathcal{Q}_m$ extends smoothly to $\left(\mathcal{T}\cup B\right)_x$. We now show that from this follows that $\mathcal{R}=L^{m+1} \mathcal{Q}_{m+1}$ for some $\mathcal{Q}_m\in C^\infty(\left(\mathcal{T}\cup B\right)_x)$. \vspace{10pt}

    Applying two vertical derivatives to the product $\mathcal{R}=L^m \mathcal{Q}_m$ yields
    \begin{align}
    {\mathcal{R}}_{\cdot\mu\cdot\nu}=m\left(m-1\right)L^{m-2}\mathcal{Q}_m\,L_{\cdot\mu}L_{\cdot\nu}+mL^{m-1}\mathcal{Q}_m\,L_{\cdot\mu\cdot\nu}+mL^{m-1}L_{\cdot\mu}\mathcal{Q}_m{}_{\cdot\nu}+mL^{m-1}L_{\cdot\nu}\mathcal{Q}_m{}_{\cdot\mu}+L^m\mathcal{Q}_m{}_{\cdot\mu\cdot\nu}.
    \end{align}
    The general identities $L_{\cdot\mu\cdot\nu}=2g^{(n)}_{\quad\mu\nu}$ and
    $L_{\cdot\mu}=\frac{2}{n-1}g^{(n)}_{\quad\mu\rho}\y^\rho$ make it so that
    \begin{equation}
        \begin{split}
             g^{(n)\,\mu\nu}{\mathcal{R}}_{\cdot\mu\cdot\nu}  =g^{(n)\,\mu\nu} & \left(  \frac{4m\left(m-1\right)}{\left(n-1\right)^2}L^{m-2}\mathcal{Q}_m{}\,g^{(n)}_{\quad\mu\rho}\y^\rho g^{(n)}_{\quad\nu\sigma}\y^\sigma+2mL^{m-1}\mathcal{Q}_m{}\,g^{(n)}_{\quad\mu\nu}+\right. \\ 
            &\;\;+\left.\frac{2m}{n-1}L^{m-1}g^{(n)}_{\quad\mu\rho}\y^\rho\mathcal{Q}_m{}_{\cdot\nu}+\frac{2m}{n-1}L^{m-1}g^{(n)}_{\quad\nu\sigma}\y^\sigma\mathcal{Q}_m{}_{\cdot\mu}+L^m\mathcal{Q}_m{}_{\cdot\mu\cdot\nu}\right).
        \end{split}
        \label{recursion}
    \end{equation}
    We use further identities for the different terms appearing in \eqref{recursion}: for the first one, $g^{(n)}_{\quad\mu\nu}\y^\mu\y^\nu=\frac{n\left(n-1\right)}{2}L$; for the second one, $ g^{(n)\,\mu\nu}g^{(n)}_{\quad\mu\nu}=4$; for the third and fourth ones, $\y^\mu\mathcal{Q}_m{}_{\cdot\mu}=\left(2-mn\right)\mathcal{Q}_m{}$. (Namely, Euler's theorem for the $\left(2-mn\right)$-homogeneous $\mathcal{Q}_m{}=\frac{\mathcal{R}}{F^{mn}}$.) All in all, \eqref{recursion} becomes
    \begin{equation}
    \begin{split}
             g^{(n)\,\mu\nu}{\mathcal{R}}_{\cdot\mu\cdot\nu}  = & \left(  \frac{2m\left(m-1\right)n}{n-1}L^{m-1}\mathcal{Q}_m{}+8mL^{m-1}\mathcal{Q}_m{}\right. \\ 
            &\;\;+\left.\frac{4m\left(2-mn\right)}{n-1}L^{m-1}\mathcal{Q}_m{}+L^mg^{(n)\,\mu\nu}\mathcal{Q}_m{}_{\cdot\mu\cdot\nu}\right).
        \end{split}
    \end{equation}
    Substituting this into the Finsler gravity equation \eqref{eq:Fgrav1} and rearranging gives
    \begin{align}
    \left[ \frac{2(n-2)}{n-1}-6+\frac{n}{2}\left(\frac{2m\left(m-1\right)n}{n-1}+8m+\frac{4m\left(2-mn\right)}{n-1}\right)\right]\mathcal{R}=-\frac{nF^n}{2 }L^{m}g^{(n)\,\mu\nu}\mathcal{Q}_m{}_{\cdot\mu\cdot\nu},
    \end{align}
    which can be simplified to
    \begin{align}
    \left( \frac{\left(3m-m^2\right)n^2-4n+2}{n-1}\right)\mathcal{R}=-\frac{n}{2 }L^{m+1}g^{(n)\,\mu\nu}\mathcal{Q}_m{}_{\cdot\mu\cdot\nu}.
    \end{align}
    As, among our hypotheses, the numerator on the left-hand side is not $0$, we get $\mathcal{R}=L^{m+1}\mathcal{Q}_{m+1}$ with $\mathcal{Q}_{m+1}$ being a constant times $g^{(n)\,\mu\nu}\mathcal{Q}_m{}_{\cdot\mu\cdot\nu}$. And since $g^{(n)\,\mu\nu}$ and $\mathcal{Q}_m{}$ were already in $C^\infty(\left(\mathcal{T}\cup B\right)_x)$, so is $\mathcal{Q}_{m+1}$. This observation concludes our proof by induction.
\end{proof}
For real applicability of Lemma \ref{recursion lemma}, we need to characterize when the condition on $n$ and $m$ is satisfied.
\begin{lemma} \label{forbidden values}
    For $n\in\mathbb{R}\setminus\left\{0,1\right\}$, the following are equivalent:
    \begin{enumerate} [label=\roman*)]
        \item  \label{forbidden values 1} $\left(3m-m^2\right)n^2-4n+2\neq 0$ for all $m\in\mathbb{N}\cup\left\{0\right\}$. 

        \item  \label{forbidden values 2} $\frac{3+\sqrt{
    9-\frac{16}{n}+\frac{8}{n^2}
    }}{2}\notin\mathbb{N}$.
    \end{enumerate}
\end{lemma}
\begin{proof}
    We show that \ref{forbidden values 1} does not hold if and only if \ref{forbidden values 2} does not hold.

    \ref{forbidden values 1} does not hold if and only if $\left(3m-m^2\right)n^2-4n+2 =  0$ for some $m\in\mathbb{N}\cup\left\{0\right\}$. Solving this equation yields
    \begin{align}
    m_+:=\frac{3}{2}+\frac{\sqrt{
    9-\frac{16}{n}+\frac{8}{n^2}
    }}{2}\in\mathbb{N}\cup\left\{0\right\}\qquad\hbox{\rm or }\qquad  m_-:=\frac{3}{2}-\frac{\sqrt{
    9-\frac{16}{n}+\frac{8}{n^2}
    }}{2}\in\mathbb{N}\cup\left\{0\right\}.
    \end{align}
    
    Observe that if $m_-\in\mathbb{N}\cup\left\{0\right\}$, then $m_+\geq m_-\geq 0$ and $m_++m_- =3$. Hence, $m_- \leq 3$ and $m_-$ can only take the values $0$, $1$, $2$ or $3$; consequently, $m_+ = 3 - m_-$ would take the values $3$, $2$, $1$ or $0$. This means that $m_+\in\mathbb{N}\cup\left\{0\right\}$ follows from $m_-\in\mathbb{N}\cup\left\{0\right\}$. Because of this, the negation of \ref{forbidden values 1} becomes equivalent to $m_+$ being a natural number or zero, namely  
    \begin{align}
        \frac{3+\sqrt{9-\frac{16}{n}+\frac{8}{n^2}}}{2}\in\mathbb{N}\cup\left\{0\right\}.
    \end{align}
    
    But this quantity cannot ever be $0$, as $\sqrt{9-\frac{16}{n}+\frac{8}{n^2}}\geq 0$. In this way, we have proved that \ref{forbidden values 1} does not hold if and only if $m_+\in\mathbb{N}$, which was our goal.
\end{proof}

\begin{remark} 
    This most precise and rather complex condition can be understood in a simplfied way by observing that
\begin{align}
    n>1\Longrightarrow \frac{3+\sqrt{9-\frac{16}{n}+\frac{8}{n^2}}}{2}\notin\mathbb{N}.
\end{align}
So, $ n>1$ is an easy condition under which the main theorem below is applicable. There are other straightforward conditions that guarantee that $\frac{3+\sqrt{9-\frac{16}{n}+\frac{8}{n^2}}}{2}\notin\mathbb{N}$, such as $n<-\frac{1+\sqrt{3}}{2}$, but we discard this in Theorem \ref{thm:analyticity main}. The reason is that $n>0$ is required there, in order for $F^n$ to be $0$, and not $\infty$, at the boundary subset $B$.
\end{remark}

\begin{theorem} \label{thm:analyticity main}  
Let $F$ be a solution of the Finsler gravity vacuum equation with $g^{(n)\mu\nu}\mathcal{P}_{\mu\nu} =0$, namely \eqref{eq:Fgrav1}. Suppose that condition {\bf (b)} in Section \ref{ssec:lightlike regularity} holds. If, for a specific $x\in M$, the set $B\cap\partial(\mathcal{T}_x)$ is non-empty and $\left.\mathcal{R}\right|_{\left(\mathcal{T}\cup B\right)_x}$ is (real-)analytic, then 
\begin{equation}
    \left.\mathcal{R}\right|_{\left(\mathcal{T}\cup B\right)_x}=0.
    \label{R = 0 again}
\end{equation}

\end{theorem}
\begin{proof}
Following \cite[Th. 2.3]{villasenor:2024}, we shall use our Lemma \ref{recursion lemma} to show the vanishing of $\mathcal{R}$. Given any $m\in\mathbb{N}\cup\left\{0\right\}$, compute the $m$-th vertical derivatives of \eqref{thesis recursion lemma} in its form $\mathcal{R}=\left(F^n\right)^{m+1} \mathcal{Q}_{m+1}$ to obtain
\begin{align}
    \mathcal{R}_{\cdot \mu_1\ldots\cdot \mu_m}=F^n\,T_{\mu_1\ldots \mu_m}\,,
\end{align}
where $T_{\mu_1\ldots \mu_m}\in C^\infty(\left(\mathcal{T}\cup B\right)_x)$ (in particular, it does not contain any negative power of $F$). Then, provided that $ B\cap\partial (\mathcal{T}_x)$ is non-empty, there exist at least one $\dot x_0\in B\cap\partial (\mathcal{T}_x)$; there, the following holds:
\begin{align}
    \mathcal{R}_{\cdot \mu_1\ldots\cdot \mu_m}(\y_0)=F(\y_0)^n \,T_{\mu_1\ldots \mu_m}(\y_0)=0\cdot T_{\mu_1\ldots \mu_m}(\y_0)=0.
\end{align}
This is true for any $m\in\mathbb{N}\cup\left\{0\right\}$: the full (vertical) Taylor series of $\mathcal{R}$ vanishes around $\y_0$. By the assumed analyticity, $\mathcal{R}=0$, in principle in some neighborhood of $\y_0$ in $\left(\mathcal{T}\cup B\right)_x$, but afterwards in this whole set due to the principle of unique analytic continuation and the connectedness of $\mathcal{T}_x$. This is precisely \eqref{R = 0 again}.
\end{proof}
\begin{remark}
    We emphasize that we only demand analyticity of the \emph{restriction} of $\mathcal{R}$ to the particular causal cone where one wants to apply Theorem \ref{thm:analyticity main}. In \cite[Def. 5.6]{villasenor:2022}, when demanded at all $x\in M$ (and on the whole $\overline{\mathcal{T}_x}=\mathcal{T}_x\cup\partial(\mathcal{T}_x)$), this notion was called \emph{fiberwise analyticity}. But even this would be weaker than the analyticity of $\mathcal{R}$ as a function on the tangent bundle, i.e., the full $\mathcal{R}(x,\y)$.

    As examples, analyticity of $\left.\mathcal{R}\right|_{\left(\mathcal{T}\cup B\right)_x}$ holds if:
    \begin{itemize}
        \item $\left.\mathcal{R}\right|_{\mathcal{T}_x}$ is quadratic, that is, a homogeneous polynomial $p(\y^0,\y^1,\y^2,\y^3)$ of degree $2$. In particular, when $F$ is of Berwald type.
        
        \item $\left.\mathcal{R}\right|_{\mathcal{T}_x}$ is a rational function $\frac{p(\y)}{q(\y)}$, where $q(\y)$ and $p(\y)$ are homogeneous polynomials of degrees, resp., $s$ and $s+2$ such that $q(\y)$ has no zeros on $\left(\mathcal{T}\cup B\right)_x$.
        
        \item $\left.\mathcal{R}\right|_{\mathcal{T}_x}$ is some combination of rational functions and real powers constructed from $(\y^0,\ldots,\y^3)$ such that the resulting degree is $2$ and the relevant quantities never vanish on $\left(\mathcal{T}\cup B\right)_x$.
    \end{itemize}

As a concrete illustration of Finsler functions with analytic $\left.\mathcal{R}\right|_{\left(\mathcal{T}\cup B\right)_x}$, consider $m$-th root metrics \eqref{eq:mth}, for which
\begin{align}
     L(x,\y):=F(x,\y)^m= G_{\mu_1\ldots \mu_m}(x)\,\dot x^{\mu_1}\ldots\dot x^{\mu_m}.
\end{align}
From \eqref{eq:spray_Fn} and \eqref{eq:curvatures_Fn}, it is easy to see that $\mathcal{R}(x,\y)$ is always rational in $\y$, its denominator $q(\y)$ being a power of $d(\y):=\det(G_{\mu\nu\rho_1\ldots \rho_{m-2}}(x)\,\dot x^{\rho_1}\ldots\dot x^{\rho_{m-2}})$. So, all it takes for $\left.\mathcal{R}\right|_{\left(\mathcal{T}\cup B\right)_x}$ to be analytic is that
\begin{align}
  d(\y)\neq 0 \; \hbox{ \rm  in } \; \mathcal{T}_x
\end{align}
and for some $\y_0$,
\begin{align}
    d(\y)\neq 0,\quad \y_0\in\partial(\mathcal{T}_x)\;\hbox{ \rm (so } \; G_{\mu_1\ldots \mu_m}(x)\,\dot x_0^{\mu_1}\ldots\dot x_0^{\mu_m}=0\hbox{\rm )}.
\end{align}
In order to apply Theorem \ref{thm:analyticity main} to the $m$-th root metric, one would need to identify $m$-th root metrics with vanishing Landsberg term $g^{(n)\mu\nu}\mathcal{P}_{\mu\nu} = 0$.
\end{remark}

We have seen how the metric Finsler gravity equation can be reduced in case that there exists a $n$-th power of the Finsler function with certain regularity conditions on some subset of the lighlike directions. Lastly, we apply these arguments to the Palatini formulation of Finsler gravity.

\subsubsection{The Palatini Finsler gravity equation} \label{palatini theory 2}

Here, we state the Palatini version of Theorem \ref{thm:analyticity main}, in a similar fashion as how Theorem \ref{thm:main extension 1} is the Palatini version of Theorem \ref{thm:maximumprinciplemain}. Again, the proof consists in repeating the steps from the previous subsubsection, basically word by word, replacing the $F$-based curvature scalar $\mathcal{R}$ by the $\bar{\mathcal{R}}$ constructed from an independent non-linear connection $\bar{N}^\mu_\nu$.
\begin{theorem} \label{thm:analyticity_palatini}
    Let $(F,\bar{N}^\mu_{\nu})$ be a $4$-dimensional solution of the Finslerian Palatini equation \eqref{eq:Fgrav2}. Suppose that condition {\bf (b)} in Section \ref{ssec:lightlike regularity} holds. If for an $x\in M$, the set $B\cap\partial(\mathcal{T}_x)$ is nonempty and $\left.\bar{\mathcal{R}}\right|_{\mathcal{T}_x}$ extends in a (real-)analytic way to  $\left(\mathcal{T}\cup B\right)_x$, then
    \begin{equation}
        \left.\bar{\mathcal{R}}\right|_{\left(\mathcal{T}\cup B\right)_x}=0.
    \end{equation}
\end{theorem}
The case in which {\bf (b)} is satisfied for $n=2$ was already mentioned in \cite[Th. 0.5 (ii)]{villasenor:2024}, but without a detailed proof. With this, we have achieved a secondary aim: to explicitly prove the statements in \cite[\S (C2)]{villasenor:2024}.

%--------------------------------------------------------------

\section{On Berwald and Landsberg type solutions } \label{eq:uniex}
The results derived in the previous sections and stated precisely in Theorems \ref{thm:maximumprinciplemain}, \ref{thm:main extension 1}, \ref{thm:analyticity main} and \ref{thm:analyticity_palatini} are particularly useful for Finsler functions of Landsberg and Berwald type \eqref{FinslerLadder}. For them, the Finsler gravity equations in the metric and the Palatini formalisms have exactly the form \eqref{eq:Fgrav3}, which is the one studied in the Theorems.

We comment on the consequences of our theorems for Berwald and Landsberg type Finsler functions and solve the Finsler gravity vacuum equations for the so-called cosmological unicorn geometry.

\subsection{Berwald solutions revisited}

During the last years, one main strategy to construct solutions to the Finsler gravity equations has been to focus on Berwald geometries \cite{FPP18,PHF21,Heefer:2024kfi}. These have the advantages that $\mathcal{P}_{\mu\nu}=0$ and $\mathcal{R}(x,\y)$ is quadratic in $\y$ (so it enjoys the partial analyticity required in Theorem \ref{thm:analyticity main}). However, it was not clear whether the easy sufficient condition $\mathcal{R}=0$ for vacuum solutions was also necessary. Now, as immediate consequences of our results in Theorems \ref{thm:maximumprinciplemain} and \ref{thm:analyticity main}, we know that this is the case when there exists \emph{some} $n$ such that $F^n$ is %as regular 
as stated in the lightlike regularity conditions {\bf (a)} or {\bf (b)} of Section \ref{ssec:lightlike regularity}:
\begin{corollary}
    Let $F$ be a Berwald type solution of the Finsler gravity vacuum equation \eqref{eq:Fgrav1}. Suppose that either of the lightlike regularity conditions {\bf (a)} or {\bf (b)} in Section \ref{ssec:lightlike regularity} holds, the latter, with $B\cap\partial(\mathcal{T}_x)$ nonempty at each $x\in M$. Then,
    \begin{align}
        \mathcal{R}=0.
    \end{align}
\end{corollary}
Even more, for Berwald type Finsler spacetimes, $\mathcal{R}=R_{\mu\nu}(x)\,\dot x^\mu \dot x ^\nu$, where %$R_{\mu\nu}(x)$ 
$R_{\mu\nu}$ is the Ricci tensor of the affine connection induced by the Berwald %type 
geometry. Hence, in these cases the vacuum Finsler gravity equation indeed reduces formally to the same expression as the Einstein vacuum equations
\begin{align}
    R_{\mu\nu}(x) = 0\,.
\end{align}

\subsection{Solving the cosmological unicorn field equations}

Having proven the conditions under which \eqref{eq:FGravVac_0} reduces to $\mathcal{R}=0$, in Theorems \ref{thm:maximumprinciplemain} and \ref{thm:analyticity main}, we now apply our results to determine the dynamics of the scale factor for the homogeneous and isotropic (cosmological) unicorn Finsler spacetime, found in \cite{Friedl-Szasz:2024vtu}. In coordinates $(t,r,\theta,\phi)$, it is defined by the Finsler function
\begin{align}\label{eq:FinslerUnicornFLRW}
    F_{\rm uni}
    :=  \left( \dot {t} f -  a(t) w \right)^{\frac{ f}{f-1}} \left( a(t) w - \dot {t} \right)^{-\frac{1 }{f-1}}\,,
\end{align}
where
\begin{align}\label{eq:w}
    w^2 = \frac{\dot r^2}{1- \kappa r^2} + r^2 (\dot \theta^2 + \sin^2\theta\,\dot\phi^2)
\end{align}
is the maximally symmetric metric on a $3$-dimensional manifold, also known as the spatial part of the Friedmann-Lema\^itre-Robertson-Walker (FLRW) geometry; $a(t)$ is the cosmological scale factor, which physically determines the expansion of the universe; and $f\in\mathbb{R}$ is a parameter.

For any choice of $f$, this Finsler function is of Landsberg type, i.e. its Landsberg tensor $P_{\mu\nu\rho}$ vanishes, and therefore the Finsler gravity equation is reduced to \eqref{eq:Fgrav1}. However, only for $f<0$ does $F_{\rm uni}$ give a Finsler spacetime with a well-defined causal structure \cite{Friedl-Szasz:2024vtu}. For $f=-1$, the geometry becomes exactly the pseudo-Riemannian FLRW one.

We now demonstrate that there exist powers of $F_{\rm uni}$ to which Theorem \ref{thm:analyticity main} applies, further reducing the gravitational field equation to $\mathcal{R}=0$.

\subsubsection{Reduction of the equation}
In order to apply Theorem \ref{thm:analyticity main}, we need to study the non-degeneracy of  $g^{(n)}$ and the smoothness of $F_{\rm uni}^{\; n}$ itself. 

As the parameter $f$ must be negative, it is convenient to introduce $Q^2:=-f$, and to consider the $n$-th power of $F_{\rm uni}$ with $n=(Q^2+1)  N $:
\begin{align}\label{eq:FinslerUnicornFLRW^n}
    F_{\rm uni}^{(Q^2+1) N }
    =  \left( - \dot {t} Q^2 -  a(t) w \right)^{Q^2  N } \left( a(t) w - \dot {t} \right)^{ N }\,,
\end{align}
where $ N >0$ is another arbitrary parameter. The determinant of the corresponding metric $g^{\left((Q^2+1) N \right)}$ is, using \cite[Eq. (67)]{Friedl-Szasz:2024vtu} and  \eqref{rel determinants},
\begin{align}\label{eq:detFinslerUnicornFLRW^n2}
    \det g^{\left((Q^2+1) N \right)}_{\mu\nu} =  \frac{a^2 \det w_{ij}}{16} (1+Q^2)^2 Q^2  N ^2 \left(1- N (1+Q^2)\right) \left( - \dot {t} Q^2 -  a(t) w\right)^{4 (Q^2  N -1)}\left( a(t) w - \dot {t}\right)^{4 ( N -1)}\,.
\end{align}

At each spacetime point $x$, the Finsler function \eqref{eq:FinslerUnicornFLRW} defines two light cones via $F_{\rm uni}=0$, meaning that
\begin{align}\label{eq:unilight}
    \left( a(t) w - \dot {t} \right) = 0 \quad \textrm{or} \quad ( - \dot {t} Q^2 -  a(t) w) = 0\,.
\end{align}

The important and striking observation is that for each of the light cones, we can choose an $ N $ such that $\det g^{\left((Q^2+1) N \right)}_{\mu\nu}$ is non-vanishing on that cone, meaning that $g^{\left((Q^2+1) N \right)}$ is non-degenerate, while $F_{\rm uni}^{(Q^2+1) N }$ becomes smooth. Thanks to this, Theorem \ref{thm:analyticity main} will apply by choosing $B$ to be the appropriate cone in each of the cases:
\begin{enumerate}
    \item Consider the first light cone, $\dot t=a(t) w$, that can be interpreted as future pointing lightlike directions, as long as $a(t) w >0$ and thus $\dot t>0$. The determinant $\det g^{\left((Q^2+1) N \right)}_{\mu\nu}$ in \eqref{eq:detFinslerUnicornFLRW^n2} is non-vanishing and non-diverging there for $ N =1$, for which
    \begin{align}\label{eq:FinslerUnicornFLRW^n2}
        F_{\rm uni}^{(Q^2+1)}
        =  \left( - \dot {t} Q^2 -  a(t) w \right)^{Q^2 } \left( a(t) w - \dot {t} \right)\,,
    \end{align}
    which clearly is smooth ($C^\infty$) on the causal cone under consideration. So, in this case, choose $n=Q^2+1$ and $B=\left\{(x,\y)\in \partial\mathcal{T}\colon\,\dot t=a(t) w \right\}$.

    \item The same argument holds when we consider the second light cone, $\dot t= - a(t) w/Q^2$, that can be interpreted as past pointing lightlike directions as $\dot t < 0$ as long as $a(t) w >0$, for another choice of $ N $. The determinant \eqref{eq:detFinslerUnicornFLRW^n2} is non-vanishing and non-diverging there for $ N =1/Q^2$, for which 
    \begin{align}\label{eq:FinslerUnicornFLRW^n3}
        F_{\rm uni}^{(Q^2+1)/Q^2}
        =  \left( - \dot {t} Q^2 -  a(t) w \right) \left( a(t) w - \dot {t} \right)^{1/Q^2}.
    \end{align}
    Analogously as the previous case, this power of $F$ is smooth ($C^\infty$) on the second causal cone, so now choose $n=(Q^2+1)/Q^2$ and $B=\left\{(x,\y)\in \partial\mathcal{T}\colon\,\dot t=- a(t) w/Q^2 \right\}$.
\end{enumerate}
For both light cones, we have shown that\footnote{ There is the subtlety that $F^{\; n}_{\rm uni}$ is just continuous, not $C^\infty$, at a line inside each timelike cone: $\left\{w=0\right\}=\mathrm{Span}\,\partial_t$. This however, does not pose any obstruction on our conclusion, which can be seen as follows. One, technically, regards these lines as part of the topological boundary $\partial\mathcal{T}$, despite $F_{\rm uni}$ not being $0$ on them as prescribed by Definition \ref{finsler lagrangian def}. Then, one can check, easily but tediously, that this detail never affects the proofs of Lemma \ref{recursion lemma} and Theorem \ref{thm:analyticity main}, nor the development of this subsection (the lines are never chosen as part of $B$). The conclusion \eqref{eq:R=0_for_unicorn} stands. } the lightlike regularity condition {\bf (b)} of Section \ref{ssec:lightlike regularity} holds. We can apply Theorem \ref{thm:analyticity main} at all points $x$, since at any of them, the chosen $B$ touch the boundary $\partial(\mathcal{T}_x)$ of the timelike cone and $\mathcal{R}$ will be an analytic function of $\dot{x}$ (one can explicitly check this in \eqref{eq:Runi}). The conclusion is that \eqref{eq:Fgrav1} always reduces to
\begin{align}
    \mathcal{R}=0\,.
   \label{eq:R=0_for_unicorn}
\end{align}

We have just demonstrated that for all causal directions, meaning all future 
pointing lightlike and timelike directions, of the unicorn Finsler geometry \eqref{eq:FinslerUnicornFLRW}, the dynamics, i.e. the scale factor, is determined simply by $\mathcal{R}=0$ instead of from the more complicated field equation \eqref{eq:Fgrav1}. 

As a last comment, we could also have regarded \eqref{eq:FinslerUnicornFLRW} as solving the Finslerian Palatini theory and carried the argument out from Theorem \ref{thm:analyticity_palatini}, as, being weakly Landsberg  ($P_\mu=0$), the unicorn geometry would then solve \eqref{eq:PalatN} and \eqref{eq:PalatR} with $\bar N^\mu{}_{\nu} = N^\mu{}_{\nu}$ and $\mathcal{\bar R}=\mathcal{R}$.

To complete this article, we derive the Ricci curvature of $F_{\rm uni}$ and solve the Finsler gravity equation from its, very simple, reduced form $\mathcal{R}=0$.

\subsubsection{Solving the Finsler gravity equation}

Spatially homogeneous and isotropic Finsler spacetimes have been studied in quite some detail \cite{Hohmann:2020mgs,Heefer:2023a,Pfeifer_2025}. Using the results from \cite[(B.28)]{Pfeifer_2025}, we can write the curvature scalar for homogeneous and isotropic Finsler functions of the type
\begin{align}
    F^2 = \dot t^2 h(t,s)^2,\quad s:=\frac{w}{\dot t}\,,
\end{align}
as
\begin{equation}
\begin{split}
\label{eq:R_unicorn_very_big}
    \frac{\mathcal{R}}{w^2} 
    &=   \frac{R^\mu{}_{\mu\nu}\dot x^\nu }{w^2}
    = -2 \kappa 
    -  \frac{3  h'''^2 \dot h'^2}{s^2 4 h''^4} 
    -  \frac{9  \bigl(h'' \dot h -  h' \dot h'\bigr)^2}{4h^2 s^2  h''^2} 
    + \frac{  \dot h' \Bigl(s h'''' \dot h' + h''' \bigl(2 \dot h' + 3 s \dot h''\bigr)\Bigr)}{2s^3 h''^3} 
    + \frac{   2 \ddot h' + s \ddot h''}{2s^3  h''}\\
    &+ \frac{3   \biggl(- h' h''' \dot h'^2 + 2 h' h'' \dot h' \dot h'' + h''^3 \ddot h -  h''^2 \Bigl(\dot h'^2 + h' \ddot h'\Bigr)\biggr)}{2h s^2  h''^3} 
    -  \frac{ -6 \dot h'^2 + 4 s \dot h' \bigl(2 \dot h'' + s \dot h'''\bigr) + s^2 \Bigl(\dot h''^2 + 2 h''' \ddot h'\Bigr)}{4s^4  h''^2}\,.
\end{split} 
\end{equation}
Here, the prime denotes the derivative with respect to $s$, and the dot, the derivative with respect to $t$.

The Finsler function $F_{\rm uni}$ of \eqref{eq:FinslerUnicornFLRW} can be written as
\begin{align}\label{eq:FinslerUnicornFLRW2}
    F_{\rm uni}^2
    =  \dot t^2 \left( f -  a(t) \frac{w}{\dot t} \right)^{\frac{2 f}{f-1}} \left( a(t) \frac{w}{\dot t} - 1 \right)^{-\frac{2 }{f-1}}\,,
\end{align}
which leads to the identification
\begin{align}
    h(t,s)^2 = \left( f -  a(t) s \right)^{\frac{2 f}{f-1}} \left( a(t) s - 1 \right)^{-\frac{2 }{f-1}}\,.
\end{align}
Evaluating the curvature scalar for this choice of $h$ yields
\begin{align}\label{eq:Runi}
    \mathcal{R}
    &= \frac{a''(t) \left(-2 (f+1) s\,a(t)+s^2 a(t)^2+3 f\right)w^2}{f s^2 a(t)}-\frac{(f-1)^2 a'(t)^2w^2}{2 f^2}-2 \kappa w^2\\
    &=\frac{3 \dot t^2 a''(t)}{a(t)}-\frac{2 \dot t w(f+1)  a''(t)}{f}+w^2 \left(\frac{a(t) a''(t)}{f}-\frac{(f-1)^2 a'(t)^2}{2 f^2}-2 \kappa\right)\,.
\end{align}
Most interestingly, this Ricci curvature is well defined on all lightlike and timelike directions, even at $w=0$, despite $F_{\rm uni}^{\;n}$ not being differentiable there. 

We see that the unicorn Ansatz \eqref{eq:FinslerUnicornFLRW} solves the (reduced) vacuum field equation $\mathcal{R}=0$ if and only if 
\begin{align}
    a(t) = c_1 + \frac{2 f \sqrt{\kappa}}{\sqrt{-(f-1)^2}} t\,;
\end{align}
hence, it leads to a real scale factor only for $\kappa<0$ or $\kappa=0$. This can be easily seen, as \eqref{eq:Runi} is polynomial in the variables $\dot t$ and $w$, and thus the coefficient in front of each power of these variables must vanish separately in order for the full polynomial to vanish. The $\dot t^2$-term implies that $a''(t)=0 \Rightarrow a(t)=c_1 + c_2 t$. Using this in the $w$-term determines one of the two constants of integration.

It is remarkable how our Theorems cleared the path to reach this result. We did not need to calculate any further derivatives of $\mathcal{R}$, as it would have been the case if we had limited ourselves to the original Finsler gravity vacuum equation for Landsberg geometries, \eqref{eq:Fgrav1}. In sight of \eqref{eq:R_unicorn_very_big}, one can only imagine the type of expressions that would have come about then.

%--------------------------------------------------------------
%--------------------------------------------------------------
%--------------------------------------------------------------
%--------------------------------------------------------------

\section{Discussion and conclusions}\label{sec:conc}
When one seeks to solve the vacuum Einstein equations, it is nearly trivial to show that this is equivalent to solving $R_{\mu\nu} = 0$. A simple algebraic manipulation suffices to reach this result, which strongly simplifies the search for solutions.

In Finsler gravity, the task is much more difficult. In this article, in Theorems \ref{thm:maximumprinciplemain}, \ref{thm:main extension 1}, \ref{thm:analyticity main} and \ref{thm:analyticity_palatini}, we presented conditions under which the scalar Finsler gravity equation in the metric (respectively, Palatini)  formalism reduces to the vanishing of the Finslerian Ricci curvature $\mathcal{R}=0$ (respectively, $\bar{\mathcal{R}}=0$). Basically, if there  exists a number $n\neq 1$ %(other more involved choices for $n$ are possible) 
(with some minor technical restrictions) such that the $n$-th power of the Finsler function, $F^n$, is regular enough on the light cones of the Finsler spacetime, then the reduction can be done, provided that the Landsberg term in the scalar equation vanishes. To prove this, we developed some non-trivial mathematical analysis, generalizing \cite{villasenor:2022}.

Thus, we found conditions that clarify when the non-variational equation $\mathcal{R}=0$ proposed by S. Rutz \cite{Rutz1993FinslerEinstein} is equivalent to the full action-based Finsler gravity equation. This feature was, to the best of our knowledge, missing in the literature until now.

We used these insights to ensure that for the cosmological unicorn spacetimes, it is necessary and sufficient to solve $\mathcal{R}=0$ in order to make them solutions of the Finsler gravity theory. The former amounts to the study of \eqref{eq:Runi}, which already greatly simplifies the calculation efforts.  In the end, the cosmological unicorn spacetimes that solve the Finsler gravity vacuum equation are exactly those with a scale factor $a(t)$ that is constant or linear in $t$ and vanishing or negative spatial curvature $\kappa$.

As for the future, our findings allow one to search for Finsler spacetimes that are Ricci-flat while being sure that this is equivalent to solving the full vacuum Finsler gravity equation. In an article in preparation, we will apply the techniques developed here to the theory including matter (for example kinetic gases) and with non-zero Landsberg term. Then, one does not have a reduction of the Finsler gravity equation, but it is possible to draw conclusions about the general behavior of the Finslerian Ricci curvature from the additional terms. This will be a first step in the direction of generalizing the famous weak, strong and null energy conditions, which are important tools in relating properties of physical matter to curvature properties of  spacetime to Finsler gravity. In particular, we will study the equivalence between a, to be identified, Finslerian strong energy condition and the timelike convergence condition.

\acknowledgments{C. P. acknowledges support by the excellence cluster QuantumFrontiers of the German Research Foundation (Deutsche Forschungsgemeinschaft, DFG) under Germany's Excellence Strategy -- EXC-2123 QuantumFrontiers -- 390837967 and was funded by the Deutsche Forschungsgemeinschaft (DFG, German Research Foundation) - Project Number 420243324.

F. F. V. was partially supported by the project PID2024-156031NB-I00 funded by MICINN/AEI and FEDER.}

\appendix

\section{Technical details on analysis and topology}\label{app:TechDet}
In this appendix, we clarify some technicalities concerning conic sets and functions on them. We need these technical facts, for instance, in the proof of Lemma \ref{continuous extension} (around equation \eqref{eq:prfL1}) or during the development of Section \ref{ssec:divbyzero}.

We start with the open and conic $\mathcal{T}\subseteq TM\setminus\mathbf{0}$ that serves as domain of our Finsler spacetime functions. To produce further conic sets, one can take its closure $\overline{\mathcal{T}}$ (say, defined as the set of limit points of $\mathcal{T}$) and boundary $\partial \mathcal{T}=\overline{\mathcal{T}}\setminus\mathcal{T}$ in the topological space $TM\setminus\mathbf{0}$. Their \emph{fibers} at $x\in M$ are
\begin{align}
    \mathcal{T}_x:=\mathcal{T}\cap T_xM\setminus\left\{0\right\},\qquad\left(\overline{\mathcal{T}}\right)_x:=\overline{\mathcal{T}}\cap T_xM\setminus\left\{0\right\},\qquad\left(\partial\mathcal{T}\right)_x:=\partial\mathcal{T}\cap T_xM\setminus\left\{0\right\}.
\end{align}
On the other hand, it makes sense to take the closure and boundary of $\mathcal{T}_x$ in the topological space $T_xM\setminus\left\{0\right\}$. We denote them, resp., by $\overline{\mathcal{T}_x}$ and $\partial \mathcal{T}_x$ (and use them at many places, see Sections \ref{ssec:maximumprinciple} and \ref{ssec:divbyzero}). One may prove in general that 
\begin{align}
    \overline{\mathcal{T}_x}\subseteq \left(\overline{\mathcal{T}}\right)_x,\qquad \partial (\mathcal{T}_x)\subseteq\left(\partial\mathcal{T}\right)_x.
\end{align}

Across the whole article, \emph{smooth} means \emph{of class $C^\infty$}; to discuss the various differentiability notions, let $m\in \mathbb{N}\cup\left\{0,\infty\right\}$. Since $\mathcal{T}$ is open in $TM\setminus\mathbf{0}$, it is defined what it means for a function $f$ there to be of class $C^m$. Now let $S\subseteq TM\setminus\mathbf{0}$ be arbitrary (e.g., $\overline{\mathcal{T}}$ or $\mathcal{T}\cup B$ as in Section \ref{ssec:lightlike regularity}) and take $f\colon S\rightarrow\mathbb{R}$. Using standard notions \cite[p. 45]{Lee}, we say that \emph{$f$ is $C^m$ on $S$}, or that $f\in C^m(S)$, if around each $(x,\y)\in S$, there are some open $O\subseteq TM$ and some $\widehat{f}\in C^m(O)$ with $\left.f\right|_{S\cap O}=\left.\widehat{f}\right|_{S\cap O}$.

In parallel, since $\mathcal{T}_x$ is open in $T_xM\setminus\left\{0\right\}$, also a smooth manifold, $C^m(\mathcal{T}_x)$ is defined from the get-go. Let $S_x\subseteq T_xM\setminus\left\{0\right\}$ be arbitrary (e.g., $\overline{\mathcal{T}_x}$ or $\left(\mathcal{T}\cup B\right)_x$) and take $h\colon S_x\rightarrow\mathbb{R}$. We put $h\in C^m(S_x)$ if around each $\y\in S_x$, there are an open $O_x\subseteq T_xM$ and $\widehat{h}\in C^m(O_x)$ with $\left.h\right|_{S_x\cap O_x}=\left.\widehat{h}\right|_{S_x\cap O_x}$.

For these $f$ and $h$:
\begin{itemize}
    \item $f\in C^m(S)$ implies that $\left.f\right|_{S_x}\in C^m(S_x)$ for all fibers $S_x=S\cap T_xM$. But this fiberwise regularity is far from sufficient for the total $f$ to be $C^m$.

    \item Even if $f\in C^m(S)$, or $h\in C^m(S_x)$, the partial derivatives of orders $1$ to $m$, on a given coordinate system, may not be well defined. An easy condition for them to be (i.e., become independent of the chosen $\widehat{f}$ or $\widehat{h}$) is that $S$, or $S_x$, is contained in the closure of its interior, in the respective space. This is always the case in this article.

    \item When $m=0$, what one recovers are continuous functions.
\end{itemize}

In the case of $h$, there is another notion: analyticity. We say that \emph{$h$ is analytic on $S_x$} if around each $\y_0\in S_x$, there are an open $O_x\subseteq T_xM$ and a power series $\underset{m\geq 0}{\sum} a_{\mu_1\ldots\mu_{m}}\left(\y^{\mu_1}-\y_0^{\mu_1}\right)\cdot\ldots\cdot\left(\y^{\mu_m}-\y_0^{\mu_m}\right)$, convergent on $O_x$, whose sum agrees with $h$ on $S_x\cap O_x$. (See \cite{krantz} for details on multivariate analytic functions.) This exploits the fact that $T_xM$ is a vector space, so no analytic manifold structure is needed on $M$ or $TM$. At any rate, analyticity of $f\colon S\rightarrow\mathbb{R}$ with respect to these would be much stronger.

These functions, with various regularity degrees, arise in our work as extensions of others, defined on open domains such as $\mathcal{T}$ or $\mathcal{T}_x$. (See Lemma \ref{continuous extension} for continuous or $C^m$ examples, or Theorem \ref{thm:analyticity main} for analytic examples.) Since our extensions are at least continuous and our sets $S$ or $S_x$ are contained, resp., in $\overline{\mathcal{T}}$ or $\overline{\mathcal{T}_x}$, the extensions are always unique. In this situation, when the extended object is $C^m$, it produces extensions of the $m$-th derivatives of the original.

For the arbitrary $S_x\subseteq T_xM\setminus\left\{0\right\}$, it makes sense to be conic ($\y\in S_x\Rightarrow \lambda\y\in S_x$ $\forall\lambda>0$). If it is, it makes sense for $h\colon S_x\rightarrow \mathbb{R}$ to be $0$-homogeneous ($h(\lambda\y)=h(\y)$ $\forall\lambda>0$). Moreover, one can define the \emph{positive projectivization} of $S_x$: the space $\mathbb{P}^+S_x$ whose points are the half-lines
\begin{equation}
    \mathbb{P}^+\y:=\left\{\lambda\y\in T_xM\colon\;\lambda>0\right\},\qquad \y\in S_x.
\end{equation}
For $\mathcal{T}=\underset{x\in M}{\bigcup}\mathcal{T}_x$, it is extensively studied in \cite{Hohmann:2021zbt} how $\mathbb{P}^+\mathcal{T}:=\underset{x\in M}{\bigcup}\mathbb{P}^+\mathcal{T}_x$ acquires the structure of a smooth manifold that is the setting of the Finslerian calculus of variations. But we are particularly interested in the case $S_x=\overline{\mathcal{T}_x}$, where we give a technical step to prove Lemma \ref{continuous extension}:
\begin{proposition} \label{properties proj}
    If $h\in C^0(\overline{\mathcal{T}_x})$ and it is $0$-hom., then $h$ attains absolute maximum and absolute minimum on $\overline{\mathcal{T}_x}$.
\end{proposition}
\begin{proof}
    Use some linear coordinates on $T_xM$ to define the sphere $\mathbb{S}:=\left\{\left(\y^0\right)^2+\ldots+\left(\y^3\right)^2=1\right\}$. Due to $0$-homogeneity, the values of $h$ are the same as those of $\left.h\right|_{\overline{\mathcal{T}_x}\cap\mathbb{S}}$, this restriction being continuous on its domain. But $\mathbb{S}$ is well known to be compact. Since $\overline{\mathcal{T}_x}$ is closed in $T_xM\setminus\left\{0\right\}$, so is $\overline{\mathcal{T}_x}\cap\mathbb{S}$ in $\mathbb{S}$: the domain of $\left.h\right|_{\overline{\mathcal{T}_x}\cap\mathbb{S}}$ also becomes compact. It is this that guarantees the existence of $\y_+,\y_{-}\in\overline{\mathcal{T}_x}\cap\mathbb{S}$ with $h(\y_+)=\underset{\overline{\mathcal{T}_x}\cap\mathbb{S}}{\max}\,h=\underset{\overline{\mathcal{T}_x}}{\max}\,h$ and $h(\y_{-})=\underset{\overline{\mathcal{T}_x}\cap\mathbb{S}}{\min}\,h=\underset{\overline{\mathcal{T}_x}}{\min}\,h$.
\end{proof}
\bibliography{nthPowerFinsler}

\end{document}